\documentclass[9pt, a4paper]{article}
\usepackage[margin=1in]{geometry} 

\usepackage{cancel}
\usepackage{amsmath}
\usepackage{amsthm}
\usepackage{amssymb}
\usepackage{braket}
\usepackage{enumitem}   
\DeclareMathOperator*{\argmin}{arg\,min}

\usepackage{tikz}
\usepackage[dvipsnames]{xcolor} \usepackage{quantikz}
\usepackage{svg}
\usepackage{import}
\usepackage{adjustbox}

\usepackage[utf8]{inputenc}
\usepackage{hyperref}
\hypersetup{
    colorlinks=true,
    linkcolor=blue,
    filecolor=magenta,      
    urlcolor=cyan,
    citecolor=BrickRed,
    pdftitle={Overleaf Example},
    pdfpagemode=FullScreen,
}

\usepackage[sort&compress, numbers,square]{natbib}
\theoremstyle{plain}
\newtheorem{theorem}{Theorem}[section]

\newtheorem{lemma}{Lemma}[section]

\newtheorem{proposition}{Proposition}[section]

\theoremstyle{definition}
\newtheorem{definition}[theorem]{Definition}

\numberwithin{equation}{section}

\usepackage{graphicx, color}
\graphicspath{{fig/}}

\usepackage{tcolorbox}
\tcbuselibrary{minted,breakable,xparse,skins}

\definecolor{bg}{gray}{0.95}
\DeclareTCBListing{mintedbox}{O{}m!O{}}{%
  breakable=false,
  listing engine=minted,
  listing only,
  minted language=#2,
  minted style=default,
  minted options={
    linenos,
    fontsize=\small,
    numbersep=8pt,
    #1},
  boxsep=0pt,
  left skip=0pt,
  right skip=0pt,
  left=19pt,
  right=0pt,
  top=6pt,
  bottom=6pt,
  arc=5pt,
  leftrule=0pt,
  rightrule=0pt,
  bottomrule=2pt,
  toprule=2pt,
  colback=bg,
  colframe=white!70,
  enhanced,
  overlay={%
    \begin{tcbclipinterior}
    \fill[orange!20!white] (frame.south west) rectangle ([xshift=15pt]frame.north west);
    \end{tcbclipinterior}},
  #3}

\usepackage{algorithm, algpseudocode} 
\usepackage{mathrsfs} 

\usepackage{lipsum}
\usepackage{titling} 

\title{An Introduction to the \\ Quantum Approximate Optimization Algorithm}
\author{
	Alessandro Giovagnoli\\ \small
	 German Aerospace Center (DLR), Germany \\
	 \small alessandro.giovagnoli@dlr.de
}

\date{
}

\begin{document}
\DeclarePairedDelimiter\ceil{\lceil}{\rceil}

\setlength\parindent{0pt}

	\maketitle

        \vspace{-20pt}
	
	\begin{abstract}
        The Quantum Approximate Optimization Algorithm (QAOA) is a promising variational quantum algorithm introduced to tackle classically intractable combinatorial optimization problems. This tutorial offers a comprehensive, first-principles introduction to QAOA and its properties, focusing on its application to Quadratic and Polynomial Unconstrained Binary Optimization (QUBO and PUBO) problems.
        The tutorial begins by outlining variational quantum circuits and QUBO problems, focusing on their key properties and the encoding of problem constraints through quadratic penalty terms. Next, it explores the QAOA in detail, covering its Hamiltonian formulation, gate decomposition, and example applications, along with their implementation and performance results. This is followed by an analysis of the algorithm’s energy landscape, where proofs are provided for its symmetry and periodicity, and where a resulting parameter space reduction is proposed. Finally, the tutorial extends these concepts to PUBO problems, generalizing the results to higher-order Hamiltonians and discussing the associated symmetries and circuit construction.
	\end{abstract}

	\tableofcontents

\newpage
\section{Background}\label{sec:background}

    The Quantum Approximate Optimization Algorithm (QAOA) is a hybrid quantum-classical algorithm that has been proposed to find approximate solutions to combinatorial optimization problems through a Variational Quantum Circuit (VQC). We briefly introduce VQCs and a particular class of combinatorial optimization problems, namely Quadratic Unconstrained Binary Optimization (QUBO) problems, which will be the focus of the first part of this document. 
    
    \subsection{Variational Quantum Circuits}
    Variational quantum circuits have become a popular solution in the current Noisy Intermediate-Scale Quantum (NISQ) era to implement quantum variational algorithms \cite{Cerezo2021}. VQCs leverage parametrized quantum gates whose optimal parameters are learned through classical optimization techniques. 
    In general, a VQC consists of a unitary operator $U(\theta)$ that has the form
    \begin{equation}\label{eq:U-circuit-in-general}
        U(\theta) = \prod\limits_{l=1}^L U_l(\theta_l) W_l
        ,
    \end{equation}
    where every $W_l$ is a fixed unitary while every $U_l(\theta_l)$ is the parametrized unitary gate representing the $l$-th layer and $\theta_l$ the vector containing the parameters of the layer. Each layer $l$ has the form 
    \begin{equation}
        U_l(\theta_l) = \prod\limits_{j=1}^M e^{- i \theta_{lj} V_j}, 
    \end{equation}    
    with $V_j$ the Pauli matrix generating the gate. 
    Given an initial state $\psi_0$ and an observable $O$, which can represent either the Hamiltonian of a physical problem whose energy we want to minimize or the cost function of an optimization problem, we define the output state $\psi(\theta) := U(\theta) \psi_0$ and the cost function $C(\theta) := \braket{\psi(\theta) | O | \psi(\theta)}$. The optimal parameters $\theta^*$ are those that minimize the average of the cost function $C(\theta)$, namely 
    \begin{equation}
        \theta^* = \argmin_{\theta} \braket{\psi(\theta) | O | \psi(\theta)}
        .
    \end{equation}
    More generally, the input state of a VQC can be represented as a mixed state described by a density matrix $\rho$. The cost function is then given by $\mathcal{C}(\theta) = \mathrm{Tr}\left[ U(\theta)\, \rho\, U^{\dagger}(\theta) O \right]$. 
    
    The optimal parameters $\theta^*$ are usually found through a classical algorithm and a classical computer: once multiple measurements have been performed and the average value of the Hamiltonian on the final state is computed, the parameters $\theta$ are updated in order to minimize $C(\theta)$. The procedure is repeated until convergence or a threshold value is reached. 
    
    Often, the optimization strategy used for VQCs is the gradient descent. Calculating the gradient of a VQC is particularly easy due to the parameter-shift rule \cite{Wierichs2022generalparameter}, which states that 
    \begin{equation}\label{eq:parameter-shift-rule}
        \partial_{\theta_k} C(\theta) \propto C(\theta + \pi/2) - C(\theta - \pi/2),
    \end{equation}
    allowing to compute each derivative $\partial_{\theta_k} C(\theta)$ of the cost function $C(\theta)$  by evaluating the circuit twice on the same parameters $\theta$, only shifted by a constant. 

    While VQCs offer significant flexibility in encoding complex quantum states, their trainability is often challenged by phenomena such as \emph{barren plateaus}, which consist in the gradient of the cost function vanishing exponentially with system size, making optimization difficult. Other factors affecting trainability include circuit depth, entanglement structure, and the choice of parameter initialization. Understanding these limitations, alongside the \emph{expressivity} of VQCs, which is defined as their ability to approximate a wide class of unitary matrices, is crucial for designing effective and efficient quantum variational algorithms
    \cite{PhysRevA.103.032430, PRXQuantum.3.010313, Ragone2024}.
    
    \subsection{Quadratic Unconstrained Binary Optimization Problems}\label{sec:qubo}
    Quadratic Unconstrained Binary Optimization (QUBO) problems are optimization problems characterized by a quadratic cost function whose variables are not subject to any constraints and take binary values. In many real problems, the variables are actually subject to one or multiple constraints, but these can be usually rewritten in a quadratic expression and summed to the cost function, in which case they are referred to as penalties. Using penalties allows us to rewrite constrained binary optimization problems as QUBO problems \cite{1811.11538}. 

    More rigorously, a QUBO problem can be defined by a vector of binary variables $x \in \{0,1\}^n$, a real-valued matrix $Q \in \mathbb{R}^{n \times n}$ and a real-valued vector $c \in \mathbb{R}^n$ and  can be written as 
    \begin{equation}\label{eq:qubo-definition}
        \text{minimize} \;\;\; x^t Q x + c^t x
        .
    \end{equation}
    We define the cost function $\mathcal{C}(x)$ as 
    \begin{equation}
        \mathcal{C}(x) := x^t Q x + c^t x 
        ,
    \end{equation}
    where $\mathcal{C}(x)$ can contain the problem constraints absorbed in  $Q$ and $c$. 
    
    In the following, we will present a short list of QUBO problems properties, as well as different examples of how optimization constraints can be rewritten in the cost function.

    \newtheorem*{propertyblock}{Properties}
    \newtheorem*{penaltyblock}{Penalties}

    \begin{propertyblock}[of QUBO Problems]
    Quadratic unconstrained binary optimization problems are characterized by the following properties:

    \begin{enumerate}[label=(\roman*)]
        \item \textbf{Binary Variables:} Since $x_i \in \{0, 1\}$, it always holds that $x^2_i = x_i$. This means that the term $c^t x$ in Equation (\ref{eq:qubo-definition}) can be rewritten as $c^t x = x^t C x$ with $C = \text{diag}\{c_1, \dots, c_n\}$ and thus $x^t Q x + c^t x = x^t(Q + C)x = x^t Q' x$.

        \item \textbf{Symmetric Matrix:} Any QUBO can be represented with a symmetric matrix \( Q \), since $x^T Q x = x^T \left(\frac{Q + Q^T}{2}\right) x$.
        
        \item \textbf{Maximization Problems:} QUBO problems can also be seen as maximization problems by minimizing $-\mathcal{C}(x)$.
    
        \item \textbf{Trivial Solutions:} If $Q_{ij} \geq 0 \; \forall i,j$ then the optimal solution is $x^* = (0, \dots, 0)$. If $Q_{ij} < 0 \; \forall i,j$ then the optimal solution is $x^* = (1, \dots, 1)$. 

        \item \label{property:v} \textbf{Scale Invariance:} If $x^*$ is an optimal solution for $\mathcal{C}(x)$, then it is also an optimal solution for $\mathcal{C}(x)/k$, with $k > 0$, since $\argmin\limits_{x} \left( \mathcal{C}(x)\right) = \argmin\limits_{x} \left(\mathcal{C}(x) / k\right)$.

        \item \textbf{NP-Hardness:} QUBO problems are NP-hard problems since a QUBO problem can be mapped to the Max Cut problem, which is proven to be NP-hard \cite{punnen2022qubo}. Also, the space of possible solutions grows exponentially with $2^n$.

    \end{enumerate}
    \end{propertyblock}

    We also give some examples of constraints on the variables and show how they can be encoded in the cost function. Namely, we can associate to each constraint $i=1, \dots, M$ a scalar function $P_i(x)$, called penalty, such that $P_i(x) > 0$ if the constraint $i$ is not respected and $P_i(x) = 0$ otherwise. We define the coefficient $p_i \in \mathbb{R_+}$ that determines the relative weight of the penalty with respect to the unconstrained cost function $x^t Q x + c^t x$. We can then rewrite the constrained cost function as 
    \begin{equation}
        \mathcal{C}(x) = x^t Q x + c^t x + \sum_{i = 1}^M p_i P_i(x)
        .
    \end{equation}
    \begin{penaltyblock}[of QUBO Problems]
    In the following we consider combinatorial binary optimization problems of size $n$ and we show examples of constraints that can be mapped to quadratic penalties $P(x)$ such that $P(x) = 0$ if the constraint is satisfied and $P(x) > 0$ otherwise:
    \begin{enumerate}
        \item 
        $x_i + x_j \leq 1 \;\;\; \to \;\;\; P(x_i, x_j) = x_i x_j$
        .
        
        \item 
        $x_i + x_j \geq 1 \;\;\; \to \;\;\; P(x_i, x_j) = (1-x_i) (1-x_j)$
        .
        
        \item 
        $x_i = x_j \;\;\; \;\;\; \;\;\; \to \;\;\; P(x_i, x_j) = x_i(1-x_j) + (1 - x_i)x_j$ 
        .
        
        \item 
        $\sum\limits_{i \in I} x_i \leq 1
        \;\;\;\;\;\; 
        \to
        \;\;\; P(x) = \sum\limits_{i,j \in I, i \neq j} x_i x_j$
        .
        
        \item 
        $\sum\limits_{i \in I} x_i = W \;\;\;\; \to \;\;\; P(x) = \left(\sum\limits_{i \in I} x_i - W \right)^2$
        .
        
        \item 
        \label{ex:slack-penalty}
        $\sum\limits_{i=1}^n w_i x_i \leq W$ with $w_i, W \in \mathbb{Z} 
        \;\; \to \;\; 
        P(x) = 
        \left(
            W - \sum\limits_{i=1}^n w_i x_i - \sum\limits_{i=n+1}^{n+m} 2^{i-(n+1)} x_{i}
        \right)^2
        $
        ,
        
        with $m = \ceil*{\log_2(W+1)}$. To define this constraint we considered that the inequality $\sum_{i=1}^n w_i x_i \leq W$ can be written as an equality taking the difference $S$ into account. Namely $\sum_{i=1}^n w_i x_i + S = W$, with $S \in \mathbb{N}$, which corresponds to the penalty $P(x) = (W - \sum_{i} w_i x_i - S)^2$. We now rewrite $S$ in its binary base. Since the maximum value of $S$ is $W$, we need $m := \ceil*{\log_2(W+1)}$ new binary slack variables $x_{n+1}, \dots, x_{n+m}$ to write $S = \sum_{i = n+1}^{n+m} s^{i-(n+1)}x_{i}$. 

        \item 
        \label{ex:unbalanced-penalty}
        $\sum\limits_{i=1}^n w_i x_i \leq W$ with $w_i, W \in \mathbb{R} 
        \;\; \to \;\; 
        P_1(x) = \sum\limits_{i=1}^n w_i x_i - W
        $
        \;\;
        and
        \;\;
        $
        P_2(x) = 
        \left(
            \sum\limits_{i=1}^n w_i x_i - W
        \right)^2
        $
        .
     
        This is a second possibility to express the same constraint of the previous example in a quadratic penalty which extends to real coefficients. The penalty $P_2$ penalizes every solution whose sum $\sum_{i=1}^n w_i x_i$ is distant from the threshold $W$, while $P_1$ penalizes solutions whose sum is above $W$ and favors solutions whose sum is below $W$. The total penalty, given by $p_1P_1(x) + p_2P_2(x)$, is not exact as penalty (\ref{ex:slack-penalty}) but it does not require extra slack variables.
    \end{enumerate}
    \end{penaltyblock}

\section{The QAOA}\label{sec:QAOA}

    In this section we derive the quantum approximate optimization algorithm which was first introduced by E. Fahri and J. Goldstone in \cite{1411.4028} to obtain approximate solutions to combinatorial optimization problems. We first highlight its analogy with quantum annealing and then develop the QAOA algorithm for QUBO problems in full generality.
    
    \subsection{Spin Hamiltonian}
        We start by showing that the cost function 
        \begin{equation}\label{eq:cost-func-binary}
            \mathcal{C}(x) = \sum\limits_{i,j = 1}^n 
            Q_{ij}
            x_i x_j  + \sum\limits_{i=1}^n c_i x_i 
            ,
        \end{equation}
        with variables $x \in \{0,1\}^{n}$, can be rewritten as a spin Hamiltonian with variables $s \in \{-1,1\}^{n}$, where $s$ represents the values of the spin in a given direction. To do this we perform the change of variable $s_i = 2 x_i - 1$ that maps $x_i \in \{0,1\}$ to $s_i \in \{-1, 1\}$.
        
        Plugging this into the cost function by substituting $x_i = \frac{s_i + 1}{2}$ we obtain the Hamiltonian $\mathcal{H}(s) = \mathcal{C}\left(\frac{s_i + 1}{2}\right)$, namely
        \begin{equation}\label{eq:qubo-ham}
        \begin{aligned}    
            \mathcal{H}(s) &= \sum_{i, j = 1}^n  \frac{Q_{ij}}{4}(s_i + 1)(s_j + 1) + \sum_{i = 1}^n \frac{c_i}{2} (s_i + 1) 
            \\
            & =  
            \sum_{i, j = 1}^n \frac{Q_{ij}}{4} s_i s_j +  \sum_{i = 1}^n \frac{1}{2} \left( c_i + \sum_{j = 1}^n Q_{ij}   \right) s_i + \bcancel{\left( \sum_{i, j = 1}^n \frac{Q_{ij}}{4} + \sum_{i = 1}^n \frac{c_i}{2}  \right)}
            ,
        \end{aligned}
        \end{equation}
        where the last term in brackets is a constant with respect to the $s$ variables and we can thus delete it. Physically this can be interpreted by remembering that the Hamiltonian, being an energy, can always be shifted by a constant.

        We remember that, as discussed in Property (ii), $Q_{ij} = Q_{ji}$, which allows us to rewrite the sum only on the indices where $i \leq j$ provided that we count twice the coefficients where $i \neq j$ and only once the coefficients where $i = j$. We then define
        \begin{equation}
            a^{ij} := 
            \begin{cases}
                \frac{Q_{ij}}{2}     &\text{if $i \neq j$}\\
                \frac{Q_{ij}}{4}     &\text{if $i = j$}\\
            \end{cases}    
            \;\;\;\;
            \text{}
            \;\;\;\;
            b^i := \frac{1}{2} \left( c_i + \sum\limits_{j=1}^n Q_{ij} \right)
        \end{equation}
        and rewrite the Hamiltonian as
        \begin{equation}\label{eq:ham-cost-function}
            \mathcal{H}(s) 
            =
            \sum\limits_{\substack{i,j = 1 \\ i \leq j}}^n a^{ij}s_is_j + \sum\limits_{i=1}^n b^i s_i
            ,
        \end{equation}
        which has the same form as the binary cost function of Equation (\ref{eq:cost-func-binary}) but with different coefficients and variables $s \in \{-1,1\}^{n}$. 
        
        By defining $J_{ij} = - Q_{ij}/4$ and $h_\text{i} = - \frac{1}{2}(c_i + \sum_{j=1}^n Q_{ij})$ we can rewrite Equation (\ref{eq:qubo-ham}) as the Hamiltonian of the Ising model, namely
        \begin{equation}\label{eq:ising-hamiltonian}
            \mathcal{H}(s) = - \sum_{i, j = 1}^n J_{ij} s_i s_j - \sum_{i = 1}^n h_\text{i} s_i
            .
        \end{equation}
        Equation (\ref{eq:ising-hamiltonian}) shows how finding the optimal solution of a QUBO cost function is equivalent to finding the lowest energy eigenstate of a physical Hamiltonian, namely the Ising Hamiltonian. Encoding the cost function of a QUBO problem into an Ising Hamiltonian is the approach used in quantum annealing \cite{rajak2023qa}, where the optimal solution is found by physically evolving a system of spins until they result in the lowest energy configuration of the Ising Hamiltonian encoding the problem. The adiabatic theorem, which we will explore in the next section, guarantees that the evolution yields the optimal solution, provided that the change in the Hamiltonian in slow enough.

    \subsection{Adiabatic Theorem and Trotterization}

        The adiabatic theorem \cite{Born1928}
        is a result of quantum mechanics presented by Born and Fock in 1928 which states the following.
        \begin{theorem}[Adiabatic Theorem]
            A physical system remains in its instantaneous eigenstate if a given perturbation is acting on it slowly enough and if there is a gap between the eigenvalue and the rest of the Hamiltonian's spectrum.
        \end{theorem}
        We stress that by ``slowly enough'' it is intended that the energy provided to modify the system Hamiltonian is smaller than the gap between the lowest energy state and the next one. The way in which this theorem can be exploited to find the lowest eigenstate of a desired Hamiltonian is the following: starting from a simple Hamiltonian whose lowest eigenstate is known, we can vary the Hamiltonian slowly enough so that the state remains in the lowest energy configuration. At the end of the time evolution, the system will be in the lowest energy state of the target Hamiltonian. This process is visually sketched in Figure~\ref{fig:ad-theorem}.
        
        \vspace{10pt}
        \begin{figure}[ht]
          \centering
          \includesvg[width = 360 pt]{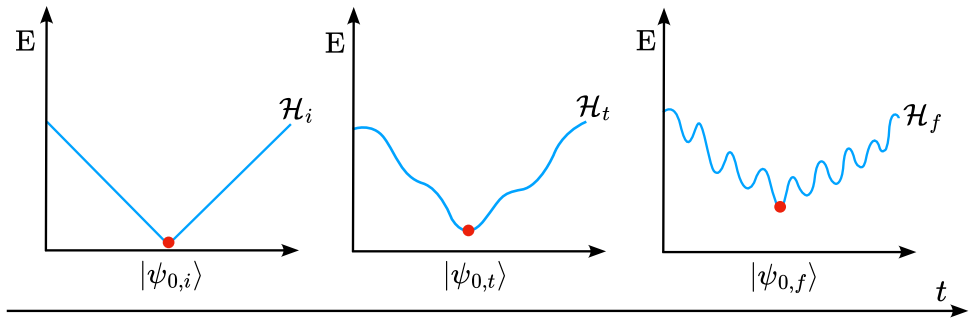}
          \caption{Pictorial representation of the Adiabatic Theorem. Starting from a system in the lowest energy state of a simple Hamiltonian, a unitary evolution is performed until the Hamiltoinian of interest is reached. According to the adiabathic theorem, now the system will be in the lowest energy state of the target Hamiltonian, provided that there is a gap between the lowest energy state and the next one.}
          \label{fig:ad-theorem}
        \end{figure}
        \vspace{10pt}
        
        To perform the unitary evolution we can define a time-dependent Hamiltonian which is composed of an initial Hamiltonian $\mathcal{H}_{\text{i}}$, whose lowest energy eigenstate is known, and a final Hamiltonian $\mathcal{H}_{\text{f}}$, the one encoding the QUBO problem. We start from Equation(\ref{eq:ham-cost-function}), which shows the Hamiltonian encoding the cost function of the optimization problem expressed with variables $s_i \in \{-1,1\}$. We now map the variables $s$ to the physical observables corresponding to the spin along the $z$-axis, which  is the Pauli-Z matrix $\sigma_z$. By mapping $s_i \to \sigma^i_z$ we obtain the final Hamiltonian $\mathcal{H}_\text{f}$, namely
        \begin{equation}\label{eq:final-qubo-ham}
        	\mathcal{H}_\text{f} 
        	:=
        	\sum_{\substack{i,j = 1 \\ i \leq j}}^n 
        	a^{ij} \sigma^i_z \sigma^j_z 
        	+
        	\sum_{i = 1}^n 
        	b^i
        	\sigma^i_z
        	.
        \end{equation}
        We stress the fact that writing $\sigma^i_z$ is a shorthand for the tensor product 
        \begin{equation}
        	\begin{aligned}
        		\sigma^i_z 
        		= 
        		\mathbb{I} \otimes \ldots  \underbrace{{\sigma_z}}_\text{$i$-th qubit}  \ldots \otimes \mathbb{I}
        		,
        	\end{aligned}
        \end{equation}
        and that, as a consequence, $\sigma^i_z \sigma^j_z$ is also a shorthand notation for
        \begin{equation}\label{eq:tensor-product-relations}
        	\begin{aligned}
        		\sigma^i_z \sigma^j_z 
        		&:=
        		\mathbb{I} \otimes \ldots \underbrace{\sigma_z}_\text{$i$-th qubit} \ldots \underbrace{\sigma_z}_\text{$j$-th qubit} \ldots \otimes \mathbb{I}
        		\;\;\;
        		\text{when}
        		\;\;\;
        		i \neq j
        		,
        		\\
        		\sigma^i_z \sigma^i_z 
        		&:= 
        		\mathbb{I} \otimes \ldots \underbrace{\left(\sigma_z\right)^2}_\text{$i$-th qubit} \ldots \otimes \mathbb{I} = \mathbb{I} \otimes \ldots \otimes \mathbb{I}
        		=
        		\mathbb{I}_{2^n}
        		.
        	\end{aligned}
        \end{equation}
        In the following we will discard the identities on the qubits where no gate is placed to make the calculations more readable. 

        Once an initial Hamiltonian $\mathcal{H}_\text{i}$, whose lowest energy eigenstate $\psi_{0, \text{i}}$ is known, has been chosen, we can define the time-evolving Hamiltonian
        \begin{equation}
            \mathcal{H}(t) :=  (1-t)\mathcal{H}_\text{i} + t \mathcal{H}_\text{f}  \;\;\; \;\;\; t \in [0,1]
            ,
        \end{equation}
        and apply the unitary evolution given by $U(t)$ that evolves the initial state $\psi_{0,\mathrm{i}}$ into the final state $\psi_{0,\mathrm{f}}$ as $\psi_{0,\mathrm{f}} = U(t)\,\psi_{0,\mathrm{i}}$. Since the Hamiltonian $\mathcal{H}(t)$ is time-dependent, as discussed in Appendix \ref{app:approximation-of-dyson}, the time evolution operator $U(t)$ is given by 
        \begin{align}
        	U(t) &= \mathcal{T}\exp\!\left\{-i\int_{t_0}^{t}
        	\mathcal{H}(s)\,ds\right\},
        \end{align}
        where we set $\hbar = 1$. $\mathcal{T}$ denotes the time-ordering operator, which arranges Hamiltonians at later times to the left. Setting $t_0 = 0$ and $t = 1$ according to the time interval $[0,1]$ that we chose for $\mathcal{H}(t)$, we get
        \begin{align}\label{eq:unitary-time-dependant}
        	U(1) 
        	= 
        	\mathcal{T} \exp\!\left\{
        	-i\int_{0}^{1} \mathcal{H}(s)\,ds
        	\right\}
        	.
        \end{align}
        In Appendix \ref{app:approximation-of-dyson} we prove that $U(1)$ can be approximated through the Trotterization at the first order, shown in Proposition \ref{prop:trotter}, as 
        \begin{align}\label{eq:approximation-of-U}
        	U(1) \approx
        	\mathcal{T}	\prod\limits_{k = 1}^{p} e^{- \mathrm{i} (1-t_k) \Delta t H_\mathrm{i}} e^{- \mathrm{i} t_k \Delta t H_\mathrm{f}}
        	,
        \end{align}
        where we defined $\Delta t = 1/p$ and $t_k = k \Delta t$. We then define the QAOA $k-$th layer operator 
        \begin{equation}\label{eq:L-operator}
        	L_k := e^{- \mathrm{i} (1-t_k) \Delta t H_\mathrm{i}} e^{- \mathrm{i} t_k \Delta t H_\mathrm{f}}
        	,
        \end{equation}
        which allows us to write the final state as 
        \begin{align}\label{eq:approximation-layer}
        	\psi_{0, \mathrm{f}} = U(1)\psi_{0, \mathrm{i}} \approx L_p \dots L_1 \psi_{0, \mathrm{i}}
        	,
        \end{align}
        where we applied the time-ordering operator $\mathcal{T}$ to arrange the operators in decreasing order. We notice from Equation \ref{eq:L-operator} that the $k-$th layer $L_k$ presents the coefficients $(1-t_k) \Delta t$ and $t_k \Delta t$ in the exponentials.  We generalize such coefficients by substituting them with free parameters, namely $(1-t_k) \Delta t \to \beta_k/2$ and $t_k \Delta t \to \gamma_k/2$, that represent the time steps of the $k$-th layer. The factor $1/2$ is due to the fact that later we will show that $\beta_k$ and $\gamma_k$ so defined can be recognized as rotation angles, the free parameters of the variational algorithm as presented in Equation (\ref{eq:U-circuit-in-general}).
        
        The reason why we introduce the free parameters is that, although the exact unitary evolution ensures that the lowest energy state $\psi_{0, \text{f}}$ is reached, we applied the approximation introduced in Equation (\ref{eq:approximation-of-U}), which facilitates the decomposition of the unitary evolution into base gates, but at the cost of losing a rigorous guarantee of convergence. Modifying the ``speed" of the time steps by optimizing the learnable parameters might thus help the evolution to converge or make it faster.
        By introducing the free parameters $\beta_k$ and $\gamma_k$ for each layer, we rewrite the $k-$th layer operator $L_k$ as
        \begin{equation}
        	L_k(\beta_k, \gamma_k) = e^{- \mathrm{i} \frac{\beta_k}{2} H_\mathrm{i}} e^{- \mathrm{i} \frac{\gamma_k}{2} H_\mathrm{f}} 
        	,
        \end{equation}
        and we rewrite the unitary operator $U(1)$ as $U(\vec{\beta}, \vec{\gamma})$, also arraning the layers in decreasing time order. We finally obtain
        \begin{equation}
        	\begin{aligned}
        		U(\vec{\beta}, \vec{\gamma})
        		& =
        		L_{p}(\beta_{p}, \gamma_{p}) \dots L_1(\beta_{1}, \gamma_1)
        		=
        		\left(
        		e^{-\mathrm{i} \frac{\beta_{p}}{2} \mathcal{H}_\text{i}}
        		e^{-\mathrm{i} \frac{\gamma_{p}}{2} \mathcal{H}_\text{f}}
        		\right)
        		\dots
        		\left(
        		e^{-\mathrm{i} \frac{\beta_{1}}{2} \mathcal{H}_\text{i}}
        		e^{-\mathrm{i} \frac{\gamma_{1}}{2} \mathcal{H}_\text{f}}
        		\right)  
        		.
        	\end{aligned}
        \end{equation}
        We define the lowest energy state of the final Hamiltonian when the evolution is truncated after $p$ layers as
        \begin{equation}\label{eq:series-of-layers}
        	\psi^{(p)}_{0, \text{f}} := L_p(\beta_p, \gamma_p) L_{p-1}(\beta_{p-1}, \gamma_{p-1}) \dots L_1(\beta_1, \gamma_1) \psi_{0, \text{i}}.
        \end{equation}
        We are then guaranteed the following two results.
        \begin{enumerate}
        	\item When $p \to \infty$ we obtain the optimal solution, since, due to Equation (\ref{eq:final-approx}) we have that
        	\begin{equation}
        		\lim\limits_{p \to \infty}\psi^{(p)}_{0, \text{f}} = \psi_{0, \text{f}}
        		.
       	\end{equation}
       	\item Adding one layer one can always increase the performances by finding a solution with a lower value of $C(\theta)$, namely
       	\begin{equation}\label{eq:adding-layer-improves}
       		\braket{\psi^{(p)}_{0, \text{f}} | \mathcal{H}_\text{f} | \psi^{(p)}_{0, \text{f}}} \leq \braket{\psi^{(p-1)}_{0, \text{f}} | \mathcal{H}_\text{f} | \psi^{(p-1)}_{0, \text{f}}}
       		.
       	\end{equation}
       	This is true since the circuit of depth $p-1$ can always be seen as a particular case of a circuit of depth $p$ where the parameters of the last layer are set to zero, namely $L_p(0,0) = \mathbb{I}$. This implies that the addition of a new layer, once the optimal parameters have been found, can only decrease the average value of the cost function or leave it unchanged, and thus produce better solutions of the optimization problem. However, we highlight the fact that the QAOA algorithm has a performance guarantee only for $p \to \infty$, so if the number of layers, and consequently the depth of the circuit, tends to infinity. When, as it must be, the sequence of layers is truncated, the evolution does not necessarily yield the optimal solution.
     \end{enumerate}
     
     \subsection{Initial Hamiltonian}
        The initial Hamiltonian that is usually chosen is the mixer Hamiltonian, namely 
        \begin{equation}\label{eq:mixer-ham}
            \mathcal{H}_\text{i} = - \sum\limits_{i = 1}^n \sigma^i_x
            ,
        \end{equation}
        where $n$ is the number of qubits and $\sigma^i_x$ the Pauli-$X$ matrix applied to the $i$-th qubit. It is easy to check that the lowest energy state of $\mathcal{H}_\text{i}$ is
        \begin{equation}
            \ket{+}^{\otimes n} = H^{\otimes n} \ket{0, \dots, 0} 
            .
        \end{equation}
        In fact, we know that the eigenvalues of $\sigma^i_x$ are $\pm 1$. From this follows that the lowest eigenvalue of $\mathcal{H}_\text{i}$ is $-n$. We can then see explicitly that $\ket{+}^{\otimes n}$ is the lowest energy eigenstate of $\mathcal{H}_\text{i}$, since
        \begin{equation}\label{eq:lowest-energy-mixer-ham}
        \begin{aligned}
                \mathcal{H}_\text{i} \ket{+}^{\otimes n}
                =
                \left( - \sum\limits_{i = 1}^n \sigma^i_x \right) \ket{+}^{\otimes n}
                = 
                - \sum\limits_{i = 1}^n  \ket{+}_1 \dots \sigma^i_x \ket{+}_i \dots \ket{+}_{n} 
                = 
                - \sum\limits_{i = 1}^n \ket{+}^{\otimes n}
                = -n \ket{+}^{\otimes n}
                ,
        \end{aligned}
        \end{equation}
        where we used the fact that $\sigma_x \ket{+} = \ket{+}$. Since $\ket{+}^{\otimes n}$ is the lowest energy eigenstate of $\mathcal{H}_\text{i}$ it can be used as a starting point for the time-evolution. This initial state can be created applying Hadamrd gates on every qubit, namely $H \otimes \dots \otimes H$.

        Having chosen $\mathcal{H}_\text{i}$ and $\mathcal{H}_\text{f}$ we will proceed with the gate decomposition of their exponential.
                
        \subsection{Gate Decomposition of the Layer Operator}
        In the following, we focus on a generic layer indexed by $p$, without assuming it to be the final layer of the circuit. Specifically, we analyze the $p$-th layer, denoted $L_p(\beta_p, \gamma_p)$, and show how it can be decomposed into a sequence of elementary quantum gates.
		We start by writing the generic $p-$th QAOA layer as
        \begin{equation}
            L_p(\beta_p, \gamma_p) = U_{\mathrm{i}}(\beta_p)  U_{\mathrm{f}}(\gamma_p) 
            ,
        \end{equation}
        where we defined 
        \begin{equation}\label{eq:Ui-Uf-definitions}
        \begin{aligned}
            U_{\mathrm{i}}(\beta_p) 
            &:=
            e^{-\mathrm{i} \frac{\beta_p}{2} \mathcal{H}_\text{i}} \\
            U_{\mathrm{f}}(\gamma_p) 
            &:= 
            e^{-\mathrm{i} \frac{\gamma_p}{2} \mathcal{H}_\text{f}}
            ,
        \end{aligned}
        \end{equation}
        and we proceed by decomposing both $U_\mathrm{i}$ and $U_\mathrm{f}$ separately.

        \subsection*{Decomposition of $U_{\mathrm{i}}(\beta_p)$}

        We start by computing the exponential of $\mathcal{H}_\text{i}$. The fact that the exponential of a Pauli matrix corresponds to a rotation operator is well known. Nevertheless, we include in the following an explicit derivation, as a similar calculation will later be extended to exponentials involving tensor products of Pauli matrices. 
        
        A direct calculation yields
        \begin{equation}\label{eq:exp-of-sigmax}
        \begin{aligned}
                U_{\mathrm{i}}(\beta_p) 
                &=
                e^{- \mathrm{i} \frac{\beta_p}{2} \sum\limits_{i = 1}^n \sigma^i_x } 
                \\
                &\stackrel{(a)}{=} 
                \prod\limits_{i = 1}^n e^{-\mathrm{i} \frac{\beta_p}{2} \sigma^i_x}
                \\
                &\stackrel{(b)}{=}  
                \prod\limits_{i = 1}^n 
                \left[
                    \sum\limits_{k = 0}^{\infty} (-1)^{k} \frac{(\beta_p/2)^{2k}}{(2k)!} \mathbb{I} 
                    -
                    \mathrm{i} \sum\limits_{k = 0}^{\infty} (-1)^k \frac{(\beta_p/2)^{2k + 1}}{(2k + 1)!} \sigma^i_x 
                \right]
                \\
                &=
                \prod\limits_{i = 1}^n
                \left[
                    \cos(\beta_p/2) \mathbb{I} 
                    -
                    \mathrm{i} \sin(\beta_p/2) \sigma^i_x 
                \right]
                \\
                &=
                \prod\limits_{i = 1}^n
                R^i_x(\beta_p)
                ,
        \end{aligned}
        \end{equation}
        where in equality ($a$) we separated the exponential function of a sum into the product of the exponential functions because $[\sigma^i_z, \sigma^j_z] = 0 \; \forall i,j$, since they act on different qubits. In equality $(b)$ we divided the sum in even and odd indices and used the fact that $(\mathrm{i} \sigma_x)^{2k} = (-1)^k \mathbb{I}$ and $(\mathrm{i} \sigma_x)^{2k+1} = (-1)^k \mathrm{i} \sigma_x$.
        In conclusion, $U_\mathrm{i}(\beta_p)$ acts as a series of $R_x$ gates on every qubit with angle $\beta_p$.         
        \subsection*{Decomposition of $U_{\mathrm{f}}(\gamma_p)$}
        
        We now compute the exponential of $\mathcal{H}_\text{f}$ and obtain
        \begin{equation}\label{eq:U-f}
        \begin{aligned}
            U_{\mathrm{f}}(\gamma_p) 
            &=
            e^{-\mathrm{i} \frac{\gamma_p}{2} \mathcal{H}_\text{f}} 
            \\
            &\stackrel{(a)}{=} 
            \exp \left\{ -\mathrm{i} \sum\limits_{\substack{i,j = 1 \\ i\leq j}}^n\frac{a_p^{ij}}{2} \sigma^i_z \sigma^j_z - \mathrm{i}\sum\limits_{i = 1}^n \frac{b_p^{i}}{2} \sigma^i_z \right\} 
            \\
            &=
            \prod\limits_{\substack{i,j = 1 \\ i\leq j}}^n \exp\left\{-\mathrm{i} \frac{a_p^{ij}}{2} \sigma^i_z \sigma^j_z \right\} \prod\limits_{i= 1}^n \exp\left\{-\mathrm{i} \frac{b_p^{i}}{2} \sigma^i_z\right\} 
            \\
            &\stackrel{(b)}{=} 
            \prod\limits_{\substack{i,j = 1 \\ i\leq j}}^n
            \exp\left\{-\mathrm{i} \frac{a_p^{ij}}{2} \sigma^i_z \sigma^j_z \right\} 
            \prod\limits_{i= 1}^n R^i_z\left(b_p^{i}\right)
            ,
        \end{aligned} 
        \end{equation}
        where in equality $(a)$ we defined the coefficients
        \begin{equation}\label{eq:a-and-b}
        \begin{aligned}    
            a_p^{ij} 
            &:=
            \gamma_p a^{ij} \\
            b_p^{i} 
            &:=\gamma_p b^i
            ,
        \end{aligned}
        \end{equation}
        and in $(b)$ we computed the terms $\exp\left\{-\mathrm{i} \frac{b_p^{i}}{2} \sigma^i_z\right\} = R^i_z\left(b_p^{i}\right)$ analogously to Equation $(\ref{eq:exp-of-sigmax})$. We now proceed to decompose the term $\exp\left\{ -\mathrm{i} \frac{a_p^{ij}}{2} \sigma^i_z \sigma^j_z \right\}$ in the last line of Equation (\ref{eq:U-f}). 
        
        We divide the cases where $i\neq j$ and $i = j$.

        \begin{enumerate}
            \item 
            When $i \neq j$, by the defintion of tensor product between operators, we have
            \begin{equation}
                \sigma^i_z \sigma^j_z
                =
                \sigma^i_z \otimes \sigma^j_z
                =
                \begin{pmatrix}
                    1 & 0\\
                    0 & -1
                \end{pmatrix}
                \otimes 
                \sigma_z
                =
                \begin{pmatrix}
                    \sigma_z & 0\\
                    0 & -\sigma_z
                \end{pmatrix}
                =
                \text{diag}\{1,-1,-1,1\}
                ,
            \end{equation}
            from which follows that
            \begin{equation} \label{eq:exp-term}
            \begin{aligned}
                \exp\left\{-\mathrm{i} \frac{a_p^{ij}}{2} \sigma^i_z \sigma^j_z \right\} 
                &=
                \exp\left\{
                    \text{diag}\left\{-\mathrm{i} \frac{a_p^{ij}}{2}, \mathrm{i} \frac{a_p^{ij}}{2}, \mathrm{i} \frac{a_p^{ij}}{2}, -\mathrm{i} \frac{a_p^{ij}}{2}\right\}
                \right\} 
                \\
                & \stackrel{(a)}{=}
                \text{diag}\left\{
                    e^{-\mathrm{i} a_p^{ij}/2}, 
                    e^{\mathrm{i} a_p^{ij}/2}, 
                    e^{\mathrm{i} a_p^{ij}/2}, 
                    e^{-\mathrm{i} a_p^{ij}/2}
                \right\}   
                \\
                &= 
                \begin{pmatrix}
                    1 & 0 & 0 & 0 \\
                    0 & 1 & 0 & 0 \\
                    0 & 0 & 0 & 1 \\
                    0 & 0 & 1 & 0 \\
                \end{pmatrix} 
                \begin{pmatrix}
                    e^{-\mathrm{i} a_p^{ij}/2} & 0 & 0 & 0 \\
                    0 & e^{\mathrm{i} a_p^{ij}/2} & 0 & 0 \\
                    0 & 0 & e^{-\mathrm{i} a_p^{ij}/2 } & 0 \\
                    0 & 0 & 0 & e^{\mathrm{i} a_p^{ij}/2} \\
                \end{pmatrix}
                \begin{pmatrix}
                    1 & 0 & 0 & 0 \\
                    0 & 1 & 0 & 0 \\
                    0 & 0 & 0 & 1 \\
                    0 & 0 & 1 & 0 \\
                \end{pmatrix} 
                \\
                &\stackrel{(b)}{=}
                \text{CNOT}(i,j)
                \left[\mathbb{I} \otimes \begin{pmatrix}
                    e^{-\mathrm{i} a_p^{ij}/2} & 0 \\
                    0 & e^{\mathrm{i}a_p^{ij}/2} \\
                \end{pmatrix}\right]
                \text{CNOT}(i,j) 
                \\
                &=
                \text{CNOT}(i,j)
                \left[
                \mathbb{I} \otimes R_z(a_p^{ij})
                \right]
                \text{CNOT}(i,j) 
                \\
            \end{aligned}
            \end{equation}
            where in (a) we used the fact that the exponential function of a diagonal operator is the diagonal operator of the exponential function. In (b) we used the matrix representation of the CNOT$(i,j)$ gate which can be directly computed considering that CNOT$(i,j) = \ket{0}_i \bra{0}_i \otimes \mathbb{I}_j + \ket{1}_i\bra{1}_i \otimes X_j$. 
    
            \item 
            When $i=j$, then $\sigma^i\otimes\sigma^i = \mathbb{I}_4$ as discussed in Equation (\ref{eq:tensor-product-relations}). From this follows that
            \begin{equation} \label{eq:exp-term-sigmaisigmai}
            \begin{aligned}
                \exp\left\{-\mathrm{i} \frac{a_p^{ii}}{2} \sigma^i_z \otimes \sigma^i_z \right\}
                &=
                \exp\left\{-\mathrm{i} \frac{a_p^{ii}}{2} \mathbb{I}_4 \right\}
                =
                e^{-\mathrm{i} \frac{a_p^{ii}}{2}} \mathbb{I}_4
                .
            \end{aligned}
            \end{equation}
            This term only adds a global phase and can thus be ignored. 
        \end{enumerate}
        Plugging the result of Equation (\ref{eq:exp-term}) in Equation (\ref{eq:U-f}) and remembering to discard the term with $i = j$ we finally obtain
        \begin{equation} \label{eq:Uf-layer}
            U_{\mathrm{f}}(\gamma_p) =
            \prod\limits_{\substack{i,j = 1 \\ i < j}}^n
            \text{CNOT}(i,j) \left[\mathbb{I} \otimes R_z(a_p^{ij})\right] \text{CNOT}(i,j)
            \prod\limits_{i= 1}^n 
            R^i_z(b_p^{i})
            .
        \end{equation}
        The general $p$-th layer operator of the QAOA circuit can finally be expressed as 
        \begin{equation}
            L_p(\beta_p, \gamma_p) = 
            \left(
                \prod\limits_{i = 1}^n
                R^i_x(\beta_p)
            \right)
            \left(
                \prod\limits_{\substack{i,j = 1 \\ i < j}}^n
                \text{CNOT}(i,j) \left[\mathbb{I} \otimes R_z(\gamma_p a^{ij})\right] \text{CNOT}(i,j)
                \prod\limits_{i= 1}^n
                R^i_z(\gamma_p b^{i})
            \right)
            .
        \end{equation}
        \subsection{The $R_{ZZ}$ Gate}
        The sequence of gates computed in Equation (\ref{eq:exp-term}) is known as $R_{ZZ}$ gate. Supposing that $i<j$, we define
        \begin{equation}\label{eq:Rzz}
            R^{i \to j}_{ZZ}(\theta) 
            :=
            \exp \left\{-i\frac{\theta}{2} \sigma^i_z \sigma^j_z\right\}
            ,
        \end{equation}
        where the superscript $i \to j$ reminds us that the CNOT controls qubit $i$ and acts on qubit $j$. From Equation (\ref{eq:exp-term}) we already know how to decompose $R^{i \to j}_{ZZ}(\theta)$, namely
        \begin{equation}
            R^{i \to j}_{ZZ}(\theta)
            =
            \text{CNOT}(i,j) 
            \left[\mathbb{I} \otimes R_z(\theta)\right] 
            \text{CNOT}(i,j)
            =
            \begin{tikzpicture}[baseline={([yshift=-.5ex]current bounding box.center)},vertex/.style={anchor=base,
            circle,fill=black!25,minimum size=18pt,inner sep=2pt}]
            \node[scale=0.7] {
                \begin{quantikz}
                & \ctrl{1} &                    & \ctrl{1} & \\
                & \targ{}  & \gate{R_Z(\theta)} & \targ{}  &
                \end{quantikz}
                };
            \end{tikzpicture}
        \end{equation}
        We now show that the $R_{ZZ}$ gate is invariant under the swap of the two qubits $i,j$. 
        
        \begin{proposition}
            The gate $R_{ZZ}$ is invariant under the swap of qubits. That is, supposing $i<j$, 
            \begin{equation}
            R^{i \to j}_{ZZ}(\theta)
            =
            \begin{tikzpicture}[baseline={([yshift=0ex]current bounding box.center)},vertex/.style={anchor=base,
            circle,fill=black!25,minimum size=18pt,inner sep=2pt}]
            \node[scale=0.7] {
                \begin{quantikz}
                \lstick{$i$} & \ctrl{1} &                    & \ctrl{1} & \\
                \lstick{$j$} & \targ{}  & \gate{R_Z(\theta)} & \targ{}  &
                \end{quantikz}
                };
            \end{tikzpicture}
            =
            \begin{tikzpicture}[baseline={([yshift=-1ex]current bounding box.center)},vertex/.style={anchor=base,
            circle,fill=black!25,minimum size=18pt,inner sep=2pt}]
            \node[scale=0.7] {
                \begin{quantikz}
                \lstick{$i$} & \targ{}   & \gate{R_Z(\theta)} & \targ{}   & \\
                \lstick{$j$}& \ctrl{-1} &                    & \ctrl{-1} &
                \end{quantikz}
                };
            \end{tikzpicture}
            = 
            R^{j \to i}_{ZZ}(\theta)
            .
            \end{equation}
        \end{proposition}

        \begin{proof}
            A direct calculation shows that
        \begin{equation}\label{eq:Rzz-invariance}
            \begin{aligned}
                R^{i \to j}_{ZZ}(\theta) 
                &=
                \text{CNOT}(i,j)
                    \left[
                    \mathbb{I} \otimes R_z(\theta)
                    \right]
                \text{CNOT}(i,j)
                \\
                &=
                \begin{pmatrix}
                    1 & 0 & 0 & 0 \\
                    0 & 1 & 0 & 0 \\
                    0 & 0 & 0 & 1 \\
                    0 & 0 & 1 & 0 \\
                \end{pmatrix} 
                \begin{pmatrix}
                    e^{-\mathrm{i} \frac{\theta}{2}} & 0 & 0 & 0 \\
                    0 & e^{\mathrm{i} \frac{\theta}{2}} & 0 & 0 \\
                    0 & 0 & e^{-\mathrm{i} \frac{\theta}{2} } & 0 \\
                    0 & 0 & 0 & e^{\mathrm{i} \frac{\theta}{2}} \\
                \end{pmatrix}
                \begin{pmatrix}
                    1 & 0 & 0 & 0 \\
                    0 & 1 & 0 & 0 \\
                    0 & 0 & 0 & 1 \\
                    0 & 0 & 1 & 0 \\
                \end{pmatrix} 
                \\
                &=
                \text{diag}\{
                    e^{-\mathrm{i} \theta/2}, 
                    e^{\mathrm{i} \theta/2}, 
                    e^{\mathrm{i} \theta/2}, 
                    e^{-\mathrm{i} \theta/2}
                \}   
                \\
                &=
                \begin{pmatrix} 
                    1 & 0 & 0 & 0 \\ 
                    0 & 0 & 0 & 1 \\ 
                    0 & 0 & 1 & 0 \\ 
                    0 & 1 & 0 & 0 
                \end{pmatrix}
                \begin{pmatrix} 
                    e^{-\mathrm{i} \frac{\theta}{2}} & 0 & 0 & 0 \\ 
                    0 & e^{-\mathrm{i} \frac{\theta}{2}} & 0 & 0 \\ 
                    0 & 0 & e^{\mathrm{i} \frac{\theta}{2}} & 0 \\ 
                    0 & 0 & 0 & e^{\mathrm{i} \frac{\theta}{2}} 
                \end{pmatrix}
                \begin{pmatrix} 
                    1 & 0 & 0 & 0 \\ 
                    0 & 0 & 0 & 1 \\ 
                    0 & 0 & 1 & 0 \\ 
                    0 & 1 & 0 & 0 
                \end{pmatrix}
                \\
                &=
                \text{CNOT}(j,i)
                    \left[
                     R_z(\theta) \otimes \mathbb{I}
                    \right]
                \text{CNOT}(j,i)
                \\
                &= R^{j \to i}_{ZZ}(\theta) 
                ,
            \end{aligned} 
            \end{equation}
            where we used the fact that
            \begin{equation}
            \begin{aligned}
                \text{CNOT}(i,j) 
                &=
                \ket{0}_i \bra{0}_i \otimes \mathbb{I}_j + \ket{1}_i\bra{1}_i \otimes X_j 
                \\
                \text{CNOT}(j,i) 
                &=
                \mathbb{I}_i \otimes \ket{0}_j \bra{0}_j + X_i \otimes
                \ket{1}_j\bra{1}_j
                ,
            \end{aligned}            
            \end{equation}
            which can be directly computed in the computational basis to obtain the matrix representations used in Equation (\ref{eq:Rzz-invariance}).
        \end{proof}

        This proposition shows that the two gates $R^{i \to j}_{ZZ}(\theta)$ and $R^{j \to i}_{ZZ}(\theta)$ can be interchanged when implementing the QAOA architecture, and, for this reason, we will drop the superscript $i \to j$ and refer more generally to the $R_{ZZ}(\theta)$ gate.

        \subsection{Scaling of the Rotation Parameters}

        We perform another fundamental manipulation by remembering Property \ref{property:v} of the QUBO problems shown in Section \ref{sec:qubo}, which states that if $s^*$ is an optimal solution for $\mathcal{H}_\text{f}(s)/k$, with $k > 0$, then it is also an optimal solution for $\mathcal{H}_\text{f}(s)$. This allows us to rescale the final Hamiltonian, which is equivalent to rescaling the coefficients $a^{ij}$ and $b^i$. The chosen rescaling factor $k > 0$ is the coefficient of the Hamiltonian with the highest absolute value, namely 
        \begin{equation}\label{eq:k}
            k = \max_{\substack{i,j,h = 1,\dots,n}} \{\, |a^{ij}|,\, |b^h| \,\}.
        \end{equation}
        Mapping the final Hamiltonian $\mathcal{H}_\text{f}$ with parameters $a^{ij}, b^i$ to $\mathcal{H}_\text{f}/k$ with parameters $a^{ij}/k, b^i/k$ does not change the optimal solution $s^*$ but dramatically facilitates the optimization of the QAOA algorithm. By performing this operation the rotation parameters are constrained on a smaller scale where the redundancies due to the periodicity in $2 \pi$ are removed, which makes the classical optimization of the cost function smoother. In Figure \ref{fig:qaoa-energylandscape} the functions  $\mathcal{H}_{\text{f}}(\beta, \gamma)$ and $\mathcal{H}_{\text{f}}(\beta, \gamma)/k$ are compared, respectively in 2B and 2C, for a particular problem instance. This fact will be examined in greater detail later in this work. 

        Taking into account the scaling of the final Hamiltonian, the general $p$-th layer of the QAOA can be finally written as
        \begin{equation}\label{eq:general-architecutre-of-layer}
            L_p(\beta_p, \gamma_p) = 
            \left(
                \prod\limits_{i = 1}^n
                R^i_x(\beta_p)
            \right)
            \left(
                \prod\limits_{\substack{i,j = 1 \\ i < j}}^n
                \text{CNOT}(i,j) 
                \left[
                    \mathbb{I} \otimes R_z\left(\frac{\gamma_p a^{ij}}{k}\right)
                \right] \text{CNOT}(i,j)
                \prod\limits_{i= 1}^n
                R^i_z\left(\frac{\gamma_p b^i}{k}\right)
            \right)
            .
        \end{equation}

    \subsection{Example Problems}
    In this section, we examine two representative combinatorial optimization problems, namely the Max Cut and the Knapsack problem, to illustrate how a general QUBO problem can be expressed within the QAOA framework and subsequently implemented as a parameterized quantum circuit.
    
    \subsection*{Example 1: Max Cut Problem}
        The Max Cut problem is often used as a benchmark example for the QAOA. We use the theoretical framework developed so far to obtain the QAOA circuit that finds its optimal solution. 
        \\
        
        \textbf{Problem Statement}: Let $\mathcal{G}(V,E)$ be an undirected graph, where $V$ is the set of vertices and $E$ the set of edges, labeled by a pair $(i,j)$ that represents the edge between node $i$ and node $j$. The Max Cut problem consists in labelling the nodes of the graph in two groups $\{0,1\}$ in such a way that encircling the nodes of the same group will lead to the maximum number of cuts through the edges. 
        \\
        
        According to the problem statement the cost function has to encode the following rule: if node $i$ and node $j$ have the same label they belong to the same group, and no cut will happen through them. In this case, we want the cost function to yield a $0$, otherwise a $1$. We then compute the quantity $x_i + x_j - 2 x_i x_j$, which is $0$ if $x_i = x_j$ and is $1$ otherwise. 

        Summing on every edge $(i,j) \in E$ and inverting the sign considering that the cost function will be minimized, we obtain
        \begin{equation}
            \mathcal{C}(x) = \sum\limits_{(i,j) \in E} -(x_i + x_j - 2x_i x_j)
            .
        \end{equation}
        The change of variables from $x_i \in \{0, 1\}$ to $s_i \in \{-1, 1\}$ yields 
        \begin{equation}\label{eq:maxcut-ham}
            \mathcal{H}_\text{f}(s) 
            =
            \mathcal{C}\left( \frac{s+1}{2} \right)
            =
            \sum\limits_{(i,j) \in E} \frac{1}{2} (s_i s_j-1) 
            =
            \sum\limits_{(i,j) \in E} \frac{1}{2} s_i s_j 
            -
            \bcancel{\sum\limits_{(i,j) \in E} \frac{1}{2}}
            ,
        \end{equation}
        which is the usual Hamiltonian for the Max Cut problem. Comparing this Hamiltonian with the one of Equation (\ref{eq:final-qubo-ham}) we recognize the coefficients
        \begin{equation}
            a^{ij} = 
            \begin{cases}
                1/2 \;\;\;\; \text{if} \;\; (i,j) \in E \\
                0 \;\;\;\;\;\;\; \text{otherwise}
            \end{cases}
            \;\;\; \text{and} \;\;\;\;\;\;\;
            b^i = 0
            \;\;\; \forall i
            .
        \end{equation}
        The scaling coefficient is $k=\frac{1}{2}$, since
        \begin{equation}
            k = \max_{\substack{i,j,h = 1,\dots,n}} \{\, |a^{ij}|,\, |b^h| \,\} = \frac{1}{2}.
        \end{equation}
        Here all the coefficients have the same values $1/2$, meaning that rescaling by $k=1/2$ will not be particularly helpful to make the optimization smoother. Indeed this can be confirmed in the plots 1B and 1C of Figure \ref{fig:qaoa-energylandscape}.

        The QAOA architecture for the Max Cut problem is then given by the generic $p-$th layer operator
        \begin{equation}\label{eq:operator-maxcut}
            L_p(\beta_p, \gamma_p) = 
            \left(
                \prod\limits_{i = 1}^n
                R^i_x(\beta_p)
            \right)
            \left(
                \prod\limits_{(i,j) \in E}^n
                \text{CNOT}(i,j) \left[\mathbb{I} \otimes R_z(\gamma_p)\right] \text{CNOT}(i,j)
            \right)
            .
        \end{equation}
        This is a particularly simple example where no linear term in the spins is present in the Hamiltonian, leading to $b^i = 0 \;\; \forall i$, and where the $a^{ij}$ are either $1/2$ or $0$. Rescaling the Hamiltonian by $k = 1/2$ thus removes every coefficient in the rotation angle of the gate $R_z$, as it can be seen in Equation (\ref{eq:operator-maxcut}). 
        
        Figure \ref{fig:qaoa-circuit} shows the implementation of the QAOA circuit with only one layer for a specific Max Cut problem given by the graph of vertices $V = \{ 0,1,2,3 \}$ and edges $E = \{(0,1), (1,2), (2,3), (3,0)\}$.
        
        \begin{figure}[ht]
        \begin{center}    
        \begin{tikzpicture}
        \node[scale=0.7] {
        \begin{quantikz}[column sep=9pt, row sep={0.8cm,between origins}]
        \lstick{$\ket{0}$} & \gate{H}\gategroup[4,steps=1,style={rounded corners, inner xsep=2pt, inner ysep=20pt}, background, label style={label position=below, anchor=north, yshift=-0.2cm}]{{\sc $\psi_{0, \mathrm{i}}$}} & & \gategroup[4,steps=16,style={rounded corners, inner xsep=1pt, inner ysep=20pt},background,label style={label position=below,anchor=north,yshift=-0.2cm}]{{\sc $L_p$}} &  \ctrl{1}\gategroup[4,steps=12,style={dashed,rounded corners, fill=blue!20, inner xsep=+2.5pt}, background]{{\sc $U_{\mathrm{f}}$}} &            & \ctrl{1} &         &                    &       &        &                  &       & \targ{}   & \gate{R_z\left(\gamma_p\right)} & \targ{}   & & \gate{R_x(\beta_p)}\gategroup[4,steps=1,style={dashed,rounded corners, fill=blue!20, inner xsep=+3pt}, background]{{\sc $U_{\mathrm{i}}$}} & \qw & \meter{} \\
    
        \lstick{$\ket{0}$}  & \gate{H} & &   & \targ{}  & \gate{R_z\left(\gamma_p\right)} & \targ{}  & \ctrl{1}  &    & \ctrl{1} &    &    &   &    &   &   & & \gate{R_x\left(\beta_p\right)} & \qw & \meter{} \\
    
        \lstick{$\ket{0}$}   & \gate{H} &  & &   &  &  & \targ{} & \gate{R_z\left(\gamma_p\right)} & \targ{}  & \ctrl{1}  &   & \ctrl{1} &  &  &  & & \gate{R_x(\beta_p)} & \qw &\meter{} \\
    
        \lstick{$\ket{0}$} & \gate{H} &  & & & &  &  &   &    & \targ{}   & \gate{R_z\left(\gamma_p\right)} & \targ{}  & \ctrl{-3} &  & \ctrl{-3} & & \gate{R_x(\beta_p)} & \qw & \meter{}
        \end{quantikz}
        };
        \end{tikzpicture}
        \caption{Example of one specific implementation of the QAOA for the Max Cut problem with a simple square graph with four vertices $V = \{ 0,1,2,3 \}$ connected by the edges $E = \{(0,1), (1,2), (2,3), (3,0)\}$. The two layers $U_{\mathrm{i}}$ and $U_{\mathrm{f}}$ together represent the main QAOA layer $L_p$, which here has been repeated only once, thus $p=1$. The layer $L_p$ can be repeated multiple times to achieve a better approximation of the lowest energy eigenstate.}
        \label{fig:qaoa-circuit}
        \end{center}
        \vspace{-15pt}
        \end{figure}
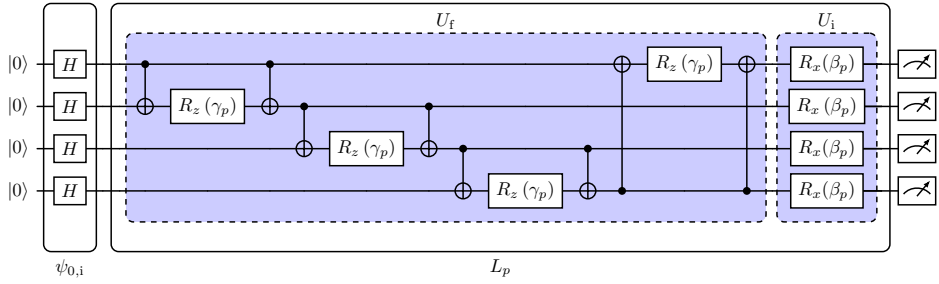

    \subsection*{Example 2: Knapsack Problem}
        The Knapsack problem is a more interesting problem for the QAOA since it also presents the linear term and a non-trivial penalty. 
        \\
        
        \textbf{Problem Statement}: A knapsack must be filled with $p$ objects. Each object $i$ has a value $v_i$ and a weight $w_i$, where both $v_i, w_i \in \mathbb{R}_+$. We want to maximize the value of the knapsack given the constraint that the weight cannot exceed the maximum weight $W \in \mathbb{R_+}$. If the $i-$th object is selected then we set $x_i=1$, otherwise $x_i=0$. The cost function and the constraint of the problem are
        
        \begin{equation}
        \begin{aligned}    
            \text{maximize}\;\;\;\;
            &\sum\limits_{i=1}^{n} v_i x_i
            \\
            \text{subject to}\;\;\;\; 
            &\sum\limits_{i=1}^{n}
            w_ix_i
            \leq
            W
            .
        \end{aligned}
        \end{equation} 
        \\

        Since we need a minimization problem, we invert the sign of the objective function. We rewrite the constraint by combining two penalties as discussed in Example (\ref{ex:unbalanced-penalty}), namely $P_1(x) = \sum\limits_{i=1}^n w_i x_i - W$ and $P_2(x) = P_1(x)^2$, and we define their corresponding relative weights $p_1$ and $p_2$. The final cost function that absorbs the penalties is
        \begin{equation}\label{eq:cost-function-of-knapsack}
        \begin{aligned}
            \mathcal{C}(x) 
            &=
            -\sum\limits_{i=1}^{n} v_i x_i
            +
            p_1
            \left(
                \sum\limits_{i=1}^n w_i x_i - W
            \right)
            +
            p_2
            \left(
               \sum\limits_{i=1}^n w_i x_i - W
            \right)^2
            ,
        \end{aligned}
        \end{equation}
    which we expand, obtaining
    \begin{equation}
    \begin{aligned}
        \mathcal{C}(x) 
        &=
        -\sum\limits_{i=1}^{n} v_i x_i
        +
        p_1 \sum\limits_{i=1}^n w_i x_i - \bcancel{p_1 W}
        +
        p_2 \sum\limits_{i,j=1}^n w_i w_j x_i x_j
        -
        2 p_2 W \sum\limits_{i=1}^n w_i x_i
        +
        \bcancel{(p_2 W)^2} 
        \\
        &=
        \sum\limits_{i,j=1}^n p_2 w_i w_j x_i x_j
        +
        \sum\limits_{i=1}^n (-v_i + p_1 w_i - 2 p_2 W w_i) x_i
        .
    \end{aligned}
    \end{equation}
    We recognize the coefficients
    \begin{equation}\label{eq:KNP-Q-c}
    \begin{aligned}
        Q_{ij} 
        = 
        p_2 w_i w_j
        \;\;\;\;
        \text{and}
        \;\;\;\;
        c_i = -v_i +p_1 w_i - 2 p_2 W w_i
        ,
    \end{aligned}
    \end{equation}
    from which follow 
    \begin{equation}\label{eq:KNP-a-b}
    \begin{aligned}
        a^{ij} = 
        \begin{cases}
            \frac{Q_{ij}}{2}     &\text{if $i \neq j$}\\
            \frac{Q_{ij}}{4}     &\text{if $i = j$}\\
        \end{cases}     
        \;\;\;\;
        \text{and}
        \;\;\;\;
        b^i
        =
        \frac{1}{2} \left( c_i + \sum_{j = 1}^n Q_{ij} \right)
        .
    \end{aligned}
    \end{equation}
    Equation (\ref{eq:KNP-a-b}) could be further expanded writing the explicit form of $Q_{ij}$ and $c_i$ but, since this does not give any useful insight into the problem, we leave it as is. 
    
    We can now rewrite the Hamiltonian as
    \begin{equation}
        \mathcal{H}_\text{f} 
        =
        \sum_{\substack{i,j = 1 \\ i \leq j}}^{n} 
        a^{ij} \sigma^i_z \sigma^j_z 
        +
        \sum_{i = 1}^{n} 
        b^i
        \sigma^i_z
        .
    \end{equation}
    The scaling coefficient $k$, namely
    \begin{equation}
        k = \max_{\substack{i,j,h = 1,\dots,n}} \{\, |a^{ij}|,\, |b^h| \,\}
        ,
    \end{equation}
    in this case is nontrivial since the Hamiltonian coefficients $a^{ij}$ and $b^h$ are not all the same, as in the Max Cut problem. Rescaling by the factor $k$ thus greatly increases the smoothness of the optimization process, as can be seen in the plots 2B and 2C of Figure \ref{fig:qaoa-energylandscape}.
    
    In conclusion, the general $p-$th layer operator for the Knapsack problem is 
    \begin{equation}\label{eq:operator-knapsack}
            L_p(\beta_p, \gamma_p) = 
            \left(
                \prod\limits_{i = 1}^{n}
                R^i_x(\beta_p)
            \right)
            \left(
                \prod\limits_{\substack{i,j = 1 \\ i < j}}^{n} 
                \text{CNOT}(i,j) 
                \left[
                    \mathbb{I} \otimes R_z\left(\frac{\gamma_p a^{ij}}{k}\right)
                \right] \text{CNOT}(i,j)
                \prod\limits_{i= 1}^{n} 
                R^i_z\left(\frac{\gamma_p b^i}{k}\right)
            \right)
            .
    \end{equation}
    In Figure \ref{fig:KSP-circuit} we show the QAOA circuit for a generic implementation of the Knapsack problem with $n=4$ variables which displays the gates pattern.

    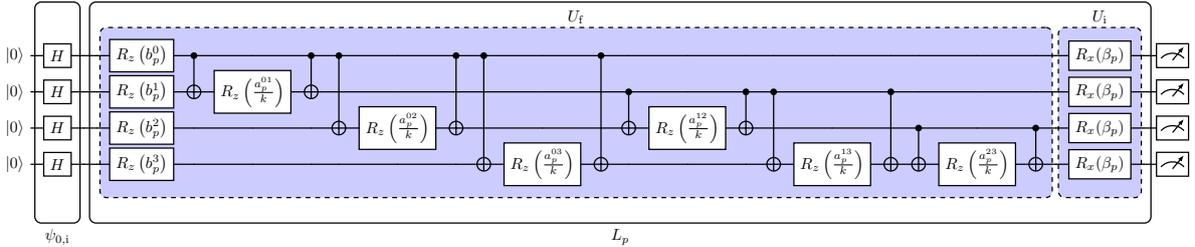
\begin{figure}[ht]
        \begin{center}       
        \begin{tikzpicture}
        \node[scale=0.6] {
        \begin{quantikz}[column sep=8pt, row sep={0.8cm,between origins}]
        \lstick{$\ket{0}$}   & \gate{H}\gategroup[4,steps=1,style={rounded corners, inner xsep=2pt, inner ysep=20pt}, background, label style={label position=below, anchor=north, yshift=-0.2cm}]{{\sc $\psi_{0, \mathrm{i}}$}} && \gategroup[4,steps=23,style={rounded corners, inner xsep=1pt, inner ysep=20pt},background,label style={label position=below,anchor=north,yshift=-0.2cm}]{{\sc $L_p$}}& \gate{R_z\left(b^0_p\right)}\gategroup[4,steps=19,style={dashed,rounded corners, fill=blue!20, inner xsep=+2.5pt}, background]{{\sc $U_{\mathrm{f}}$}} & \ctrl{1}  &                                           & \ctrl{1}  & \ctrl{2}  &                                           & \ctrl{2}  & \ctrl{3}  &                                           & \ctrl{3}  &           &                                           &           &           &                                           &           &           &                                           &          & &\gate{R_x(\beta_p)}\gategroup[4,steps=1,style={dashed,rounded corners, fill=blue!20, inner xsep=+3pt}, background]{{\sc $U_{\mathrm{i}}$}}  & & \meter{} \\
       
        \lstick{$\ket{0}$}   & \gate{H} &&& \gate{R_z\left(b^1_p\right)} & \targ{}   & \gate{R_z\left(\frac{a^{01}_p}{k}\right)} & \targ{}   &           &                                           &           &           &                                           &           & \ctrl{1}  &                                           & \ctrl{1}  & \ctrl{2}  &                                           & \ctrl{2}  &           &                                           &          &  &\gate{R_x(\beta_p)} & & \meter{} \\
     
        \lstick{$\ket{0}$}   & \gate{H} &&& \gate{R_z\left(b^2_p\right)} &           &                                           &           & \targ{}   & \gate{R_z\left(\frac{a^{02}_p}{k}\right)} & \targ{}   &           &                                           &           & \targ{}   & \gate{R_z\left(\frac{a^{12}_p}{k}\right)} & \targ{}   &           &                                           &           & \ctrl{1}  &                                           &\ctrl{1}  &  &\gate{R_x(\beta_p)} & & \meter{} \\
      
        \lstick{$\ket{0}$}   & \gate{H} &&& \gate{R_z\left(b^3_p\right)} &           &                                           &           &           &                                           &           & \targ{}   & \gate{R_z\left(\frac{a^{03}_p}{k}\right)} & \targ{}   &           &                                           &           & \targ{}   & \gate{R_z\left(\frac{a^{13}_p}{k}\right)} & \targ{}   & \targ{}   & \gate{R_z\left(\frac{a^{23}_p}{k}\right)} & \targ{}  & &\gate{R_x(\beta_p)}  & &  \meter{}
        \end{quantikz}
        };
        \end{tikzpicture}
        \vspace{-10pt}
        \caption{Example of one generic implementation of the QAOA circuit with one layer for the Knapsack problem with $n=4$ variables. The coefficients are $a_p^{ij} = \gamma_p a^{ij}$ and $b^i_p = \gamma_p b^i$, with $a^{ij}, b^i$ as defined in Equation (\ref{eq:KNP-a-b}).}
        \label{fig:KSP-circuit}
        \end{center}
        \vspace{-10pt}
    \end{figure}

\section{Energy Landscape and Symmetries}   

    In this chapter further details about the algorithm are analyzed. We will study the loss function landscape, also known as the energy landscape, the symmetries and periodicity of the QAOA architecture and a consequent way to restrict the parameter space.

    \subsection{Energy Landscape}

    Let $\psi(\vec{\beta},\vec{\gamma}) = L_p(\beta_p, \gamma_p) \dots L_1(\beta_1, \gamma_1) \psi_{0, \text{i}}$ be the output state of the QAOA circuit with target Hamiltonian $\mathcal{H}_\text{f}$, with $\vec{\beta} = (\beta_1, \dots, \beta_p)$ and $\vec{\gamma} = (\gamma_1, \dots, \gamma_p)$.

    \begin{definition}[Energy landscape]
        We call \textit{energy landscape} of the QAOA the average value of the target Hamiltonian as a function of the parameters. Namely
    \end{definition}
    \begin{equation}\label{eq:energy-landscape}
        \braket{\mathcal{H}_\text{f}}(\vec{\beta}, \vec{\gamma}) 
        =
        \braket{\psi(\vec{\beta}, \vec{\gamma})|
        \mathcal{H}_\text{f}
        |\psi(\vec{\beta}, \vec{\gamma})}
        .
    \end{equation}
    In the specific case of the QAOA architecture the energy landscape is a function of $2p$ variables, with $p$ the number of layers, but the energy landscape can be defined in general for every VQC with parameters $\theta$ as $\braket{\mathcal{H}}(\theta) = \braket{\psi(\theta)|\mathcal{H}|\psi(\theta)}$. 

    Due to the presence of rotation gates, the QAOA energy landscape is expected to be highly sinusoidal, exhibiting multiple crests and troughs, as shown in Figure \ref{fig:qaoa-energylandscape}. Here we present the energy landscape after one layer of the QAOA architecture for the Max Cut (row 1) and the Knapsack (row 2) problems given one specific choice of the coefficients. For each problem we show the three-dimensional graph of the energy landscape (column A), its two-dimensional projection (column B) and, for comparison, the same two-dimensional graph of the energy landscape when the rotation parameters have not been rescaled by the factor $k$ presented in Equation (\ref{eq:k}) (column C). The graphs have been restricted to the interval $\beta, \gamma  \in [-\pi, \pi]$ due to the periodicity of the rotation parameters.
        
    \begin{figure*}[ht]
         \centering
         \includegraphics[width = 460pt]{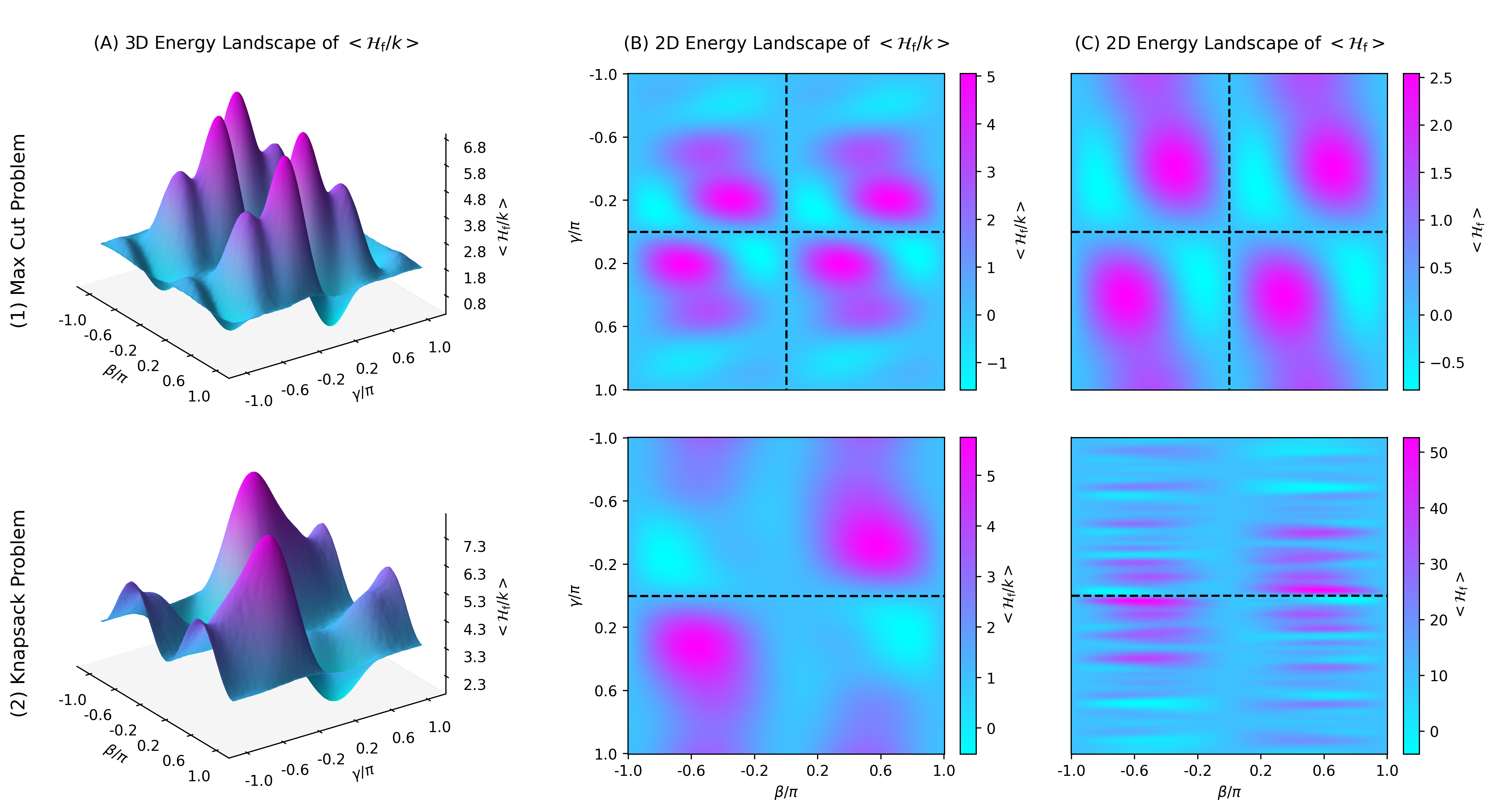}
         \caption{Plots of the energy landscape of one QAOA layer. Row 1 and row 2: plots of the energy landscape after one QAOA layer for, respectively, the Max Cut and Knapsack problem. Column A: three dimensional plot of the energy landscape of the scaled Hamiltonian $\mathcal{H}_\text{f}/k$. Column B: two dimensional plot of the energy landscape of the scaled Hamiltonian $\mathcal{H}_\text{f}/k$. Column C: two dimensional plot of the energy landscape of the Hamiltonian $\mathcal{H}_\text{f}$.}
         \label{fig:qaoa-energylandscape}
    \end{figure*}

    From Figure \ref{fig:qaoa-energylandscape} one can also see how the energy landscape of both problems presents symmetries, which are highlighted by the dotted lines in the two-dimensional graphs: the Max Cut problem energy landscape is symmetric with respect to the origin of the axis, namely the point $(\beta, \gamma) = (0,0)$, and it is periodic in the parameter $\beta$ of period $\pi$. The energy landscape of the Knapsack problem is instead only symmetric with respecto to the origin. It can be proven that this properties hold in general provided some assumptions on the form of the Hamiltonian, and this is what we will investigate in the next section.

    \subsection{Symmetry and Periodicity}
    The fact that the parameter space is symmetric with respect to the origin of the axis can be proven to always be true.

    \begin{proposition}[QAOA Symmetry for QUBO problems]\label{prop:simmetry}
        The energy landscape of the QAOA that implements a QUBO problem is symmetric with respect to the origin. Namely 
        \begin{equation}\label{eq:symmetry}
            \braket{\mathcal{H}_\mathrm{f}}(-\vec{\beta}, -\vec{\gamma}) = \braket{\mathcal{H}_\mathrm{f}}( \vec{\beta}, \vec{\gamma})
        .
        \end{equation}
    \end{proposition}

    \begin{proof}
        We show this using only one layer $L(\beta, \gamma)$, since the generalization to multiple layers is straightforward. We  first show two useful properties.

        \begin{enumerate}[label=(\roman*)]
            \item If $\psi$ is real valued, meaning that $\psi_i = \psi^*_i$ $\forall i$, then
                \begin{equation}
                    \braket{\psi,A\psi}^* 
                    =
                    \left[ \sum\limits_i \psi_i (A \psi)^*_i \right]^* 
                    =
                    \sum\limits_i \psi_i (A^* \psi)^*_i
                    =
                    \braket{\psi,A^*\psi} 
                    .
                \end{equation}
            \item  For the layer operator $L(\beta, \gamma)$, and analogously for $L^\dagger(\beta, \gamma)$, it holds that
                \begin{equation}
                    L^*(-\beta,-\gamma) 
                    =
                    \left(
                        e^{\mathrm{i}(\beta/2)H_\text{i}}
                        e^{\mathrm{i}(\gamma/2)H_\text{f}}
                    \right)^*
                    =
                    e^{-\mathrm{i}(\beta/2)H_\text{i}}e^{-\mathrm{i}(\gamma/2)H_\text{f}}
                    =
                    L(\beta, \gamma)
                    .
                \end{equation}
        \end{enumerate}

        We now prove the main claim of Equation (\ref{eq:symmetry}):
        \begin{equation}
        \begin{aligned}
            \braket{\mathcal{H}_\text{f}}(-\beta, -\gamma) 
            &=
            \braket{\psi(-\beta, -\gamma) | \mathcal{H}_\text{f} | \psi(-\beta, -\gamma)}
            \\
            &=
            \braket{
                \psi_{0, \text{i}} |
                L^\dagger(-\beta, -\gamma)  \mathcal{H}_\text{f}  L(-\beta, -\gamma)
                | \psi_{0, \text{i}}
            }
            \\
            &\stackrel{(a)}{=}
            \braket{
                \psi_{0, \text{i}} |
                L^\dagger(-\beta, -\gamma)  \mathcal{H}_\text{f}  L(-\beta, -\gamma)
                | \psi_{0, \text{i}}
            }^*
            \\
            &\stackrel{(b)}{=}
            \braket{
                \psi_{0, \text{i}} |
                \left[L^\dagger(-\beta, -\gamma)\right]^*  \mathcal{H}_\text{f}^*  L^*(-\beta, -\gamma)
                | \psi_{0, \text{i}}
            }
            \\
            &\stackrel{(c)}{=}
            \braket{
                \psi_{0, \text{i}} |
                L^\dagger(\beta, \gamma)  
                \mathcal{H}_\text{f}  L(\beta, \gamma)
                | \psi_{0, \text{i}}
            }
            \\
            &=
            \braket{\mathcal{H}_\text{f}}(\beta, \gamma) 
            ,
        \end{aligned}
        \end{equation}
        where in $(a)$ we used the fact that $\braket{\mathcal{H}_\text{f}} \in \mathbb{R}$; in $(b)$ we used property $(\text{i})$, which can be applied since $\psi_{0, \text{i}} = \ket{+}^{\otimes n}$ is real valued; in $(c)$ we used property $(\text{ii})$ on $L$ and $L^\dagger$ and the fact that $\mathcal{H}_\text{f}$ as defined in Equation (\ref{eq:final-qubo-ham}) is real.
        \end{proof}
        We now turn our attention to the periodicity of the QAOA parameter space. The periodicity of $2\pi$ in the parameters $\beta$ and $\gamma$ is straightforward, as it arises directly from their interpretation as rotation angles in the unitary operators. The periodicity in $\beta$ of period $\pi$ is instead not true in general and only holds true for QUBO problems whose final Hamiltonian does not present linear terms in the spins.

        Before proving the main result we prove a useful lemma.
        \begin{lemma}\label{lemma:commutation}
            If the final Hamiltonian $\mathcal{H}_\mathrm{f}$ of the QAOA has only quadratic terms in the spins then the operator $\sigma^1_x \dots \sigma^n_x$ commutes with every layer of the QAOA circuit. Namely 
            \begin{equation}\label{eq:lemma-main-claim}
                [\sigma^1_x \dots \sigma^n_x, L_k(\beta, \gamma)] = 0
                ,
            \end{equation}
            where $L_k$ is any arbitrary layer of the QAOA.
        \end{lemma}
        \begin{proof}
            We first show that 
            \begin{equation}\label{eq:commutator-2-sigma}
            \begin{aligned}
                \left[ \sigma^i_x \otimes \sigma^j_x, \sigma^i_z \otimes \sigma^j_z \right]  = 0,
            \end{aligned} 
            \end{equation}
            which we can directly compute, removing the $i,j$ indices for notational simplicity, as
            \begin{equation}\label{eq:commutator-proof}
                \begin{aligned}
                    \sigma_x \otimes \sigma_x 
                    \cdot 
                    \sigma_z \otimes \sigma_z
                    =
                    \sigma_x \sigma_z \otimes \sigma_x \sigma_z
                    =
                    (-\sigma_z \sigma_x) \otimes (-\sigma_z \sigma_x)
                    =
                    \sigma_z \otimes \sigma_z \cdot \sigma_x \otimes \sigma_x
                    ,
                \end{aligned}
            \end{equation}
            where we used the fact that $\sigma_x \sigma_z = -\sigma_z \sigma_x$. From Equation (\ref{eq:commutator-2-sigma}) it follows that also
            \begin{equation}\label{eq:commutator-sigma-exp}
                \left[ \sigma^i_x \otimes \sigma^j_x, \exp \left\{ \sigma^i_z \otimes \sigma^j_z \right\} \right]  = 0,
            \end{equation}
            since in general, if $[A,B] = 0$, then $[A, f(B)] = 0$, which can easily be proven by expanding $f(B)$ in Taylor series. Equation (\ref{eq:commutator-sigma-exp}) implies that  
            \begin{equation}\label{eq:commutator-sigmax-Rzz}
                \left[
                    \sigma^1_x \dots \sigma^n_x, 
                    \prod\limits_{i < j = 1}^n 
                    \exp\left\{-\mathrm{i} \frac{a_k^{ij}}{2} \sigma^i_z \sigma^j_z \right\} 
                \right] = 0
                ,
            \end{equation}
            since, for every $i<j$ we can always decompose $\sigma^1_x  \dots \sigma^i_x \dots \sigma^j_x \dots \sigma^n_x$ into two factors, namely
            \begin{equation}
            \begin{aligned}   
            \sigma^1_x  \dots \sigma^i_x \dots \sigma^j_x \dots \sigma^n_x = 
                (\sigma^1_x \dots \mathbb{I} \dots \mathbb{I} \dots \sigma^n_x) \cdot (\mathbb{I} \dots \sigma^i_x \dots \sigma^j_x \dots \mathbb{I})
                ,
            \end{aligned}
            \end{equation}
            where both commute with the term $\exp\left\{-\mathrm{i} \frac{a_k^{ij}}{2} \sigma^i_z \sigma^j_z \right\}$: the first because the qubits involved are different, and the second because of Equation (\ref{eq:commutator-sigma-exp}). By repeating this procedure for every pair of $i,j$ with $i<j$ we obtain the result of Equation (\ref{eq:commutator-sigmax-Rzz}).

            It is now easy to prove that $\sigma^1_x \dots \sigma^n_x$ commutes with $L(\beta, \gamma)$: by explicitly writing $L(\beta, \gamma)$, namely
            \begin{equation}
                L_k(\beta_k, \gamma_k) = 
                \prod\limits_{i = 1}^n 
                \exp\left\{-\mathrm{i} \frac{\beta_k}{2} \sigma^i_x\right\}
                \prod\limits_{i < j = 1}^n
                \exp\left\{-\mathrm{i} \frac{a_k^{ij}}{2} \sigma^i_z \sigma^j_z \right\} 
                ,
            \end{equation}
            we see that the operator $\sigma^1_x \dots \sigma^n_x$ commutes with every term in the products: it commutes with all the terms $\exp\left\{-\mathrm{i} \frac{\beta_k}{2} \sigma^i_x\right\}$ because they are a function of the same operator $\sigma_x$, and it commutes with the terms $\exp\left\{-\mathrm{i} \frac{a_k^{ij}}{2} \sigma^i_z \sigma^j_z \right\}$ due to Equation (\ref{eq:commutator-sigmax-Rzz}).
        
        \end{proof}

        This lemma intuitively shows that, if the Hamiltonian has only quadratic terms $s_i s_j$, with $s \in \{-1,1\}^n$, then it does not matter if we first apply one layer of the QAOA and then we flip every qubit applying $\sigma^1_x \dots \sigma^n_x$ or if we first flip the qubits and then apply the layer. Indeed every term $s_is_j$ would yield the same contribution if the sign of $s_i$ and $s_j$ is flipped since, trivially, $(-s_i)(-s_j) = s_is_j$. The same is not true if there is a linear term.
        
        We are now ready to prove the periodicity of $\beta$.
        
        \begin{proposition}[QAOA Periodicity for QUBO problems]\label{prop:period}
        If the final Hamiltonian $\mathcal{H}_\mathrm{f}$ of the QAOA has only quadratic terms in the spins then its associated energy landscape is periodic in its variables $\vec{\beta}$ with period $\pi$. Namely,
        \begin{equation}\label{eq:main-claim}
            \braket{\mathcal{H}_\mathrm{f}}(\beta_1, \dots, \beta_i + \pi, \dots \beta_{p}, \vec{\gamma}) = \braket{\mathcal{H}_\mathrm{f}}(\vec{\beta}, \vec{\gamma})
            \;\;\;
            \forall i
            .
        \end{equation}
    \end{proposition}

    \begin{proof}
        We preliminary notice that 
        \begin{equation}
            \begin{aligned}
                L_i(\beta_i + \pi, \gamma_i) 
                = 
                 e^{-\mathrm{i} \left(\frac{\beta_i}{2} + \frac{\pi}{2}\right) \mathcal{H}_\text{i}}
                e^{-\mathrm{i} \frac{\gamma_i}{2} \mathcal{H}_\text{f}}
                =
                e^{-\mathrm{i} \frac{\pi}{2} \sum\limits_{k = 1}^n \sigma^k_x}
                L_i(\beta_i, \gamma_i)
                =
                (-\mathrm{i})^n  \sigma^1_x, \dots, \sigma^n_x
                L_i(\beta_i, \gamma_i) 
                ,
            \end{aligned}
        \end{equation}
        where we can neglect the term $(-\mathrm{i})^n$ since it is a global phase and thus will not change the energy landscape. We compute
        \begin{equation}\label{eq:state-periodicity}
        \begin{aligned}
            \psi(\beta_1, \dots, \beta_i + \pi, \dots, \beta_p, \vec{\gamma})
            &=
            L_p(\beta_p, \gamma_p) 
            \dots
            L_i(\beta_i + \pi, \gamma_i)
            \dots
            L_1(\beta_1, \gamma_1) \psi_{0, \text{i}} 
            \\
            &=
            L_p(\beta_p, \gamma_p) 
            \dots
            \sigma^1_x, \dots, \sigma^n_x
            L_i( \beta_i, \gamma_i)            
            \dots
            L_1(\beta_1, \gamma_1)
            \psi_{0, \text{i}}
            \\
            &\stackrel{(a)}=
            L_p(\beta_p, \gamma_p) 
            \dots
            L_1(\beta_1, \gamma_1)
            \sigma^1_x, \dots, \sigma^n_x
            \psi_{0, \text{i}}
            \\
            &\stackrel{(b)}=
            L_p(\beta_p, \gamma_p) 
            \dots
            L_1(\beta_1, \gamma_1)
            \psi_{0, \text{i}}
            \\
            &= 
            \psi(\vec{\beta}, \vec{\gamma})
            ,
        \end{aligned}
        \end{equation}
        where in equality (a) we let the operator $\sigma^1_x, \dots, \sigma^n_x$ commute through all the layers because of Lemma \ref{lemma:commutation} and in the equality (b) we used the fact that $\psi_{0, \text{i}}$ is an eigenstate of the operator $\sigma^1_x, \dots, \sigma^n_x$ with eigenvalue 1. From the result of Equation (\ref{eq:state-periodicity}) the main claim of Equation (\ref{eq:main-claim}) directly follows.
        
    \end{proof}

     We conclude recalling that, in principle, the initial Hamiltonian $H_\text{i}$ is arbitrary, and one might think of choosing a more suitable one. We highlight that these symmetries could then be lost, since they both rely on the chosen $H_\text{i}$: the lowest state of this Hamiltonian, namely $\psi_{0, \text{i}} = \ket{+}^{\otimes n}$, is real valued and is an eigenstate of $\sigma^1_x \dots \sigma^n_x$, which, respectively, allowed us to apply the property $(\text{i})$ in Proposition \ref{prop:simmetry} and to make the operator $\sigma^1_x \dots \sigma^n_x$ disappear in the last step of Proposition \ref{prop:period}. A different $H_\text{i}$ might be endowed with different symmetries, but not necessarily the ones shown here.

    \subsection{Parameter Space Restriction}\label{sec:param-space-restriction}
    Because of these symmetries we can further restrict the parameter space of the QAOA.
    \begin{enumerate}
        \item Taking into account the periodicity of $2\pi$ for all the variables, we restrict $(\vec{\beta}, \vec{\gamma}) \in [-\pi,\pi]^p\times[-\pi, \pi]^p$. The $2p$-dimensional volume of the parameter space is $V=(2\pi)^p (2\pi)^p$.
        
        \item Proposition \ref{prop:simmetry} holds for every Hamiltonian and thus we can always restrict the optimization parameters to $(\vec{\beta}, \vec{\gamma}) \in [0,\pi]^p\times[-\pi, \pi]^p$ since $\braket{\mathcal{H}_\text{f}}(-\beta, \gamma) = \braket{\mathcal{H}_\text{f}}(\beta, -\gamma)$. 
        Taking the symmetry with respect to the origin into account, the $2p-$dimensional volume becomes $V_\mathrm{s}=(\pi)^p(2\pi)^p = V/2^p$.
        
        \item Proposition \ref{prop:period} only holds when the spin Hamiltonian presents only quadratic terms. In this case we can further restrict the parameters to $(\vec{\beta}, \vec{\gamma}) \in [0,\pi]^p\times[0, \pi]^p$, since $\braket{\mathcal{H}_\text{f}}(\beta - \pi, \gamma) = \braket{\mathcal{H}_\text{f}}(\beta, \gamma)$. The $2p-$dimensional volume taking into account the symmetry and the periodicity of $\pi$ in $\beta$ is $V_{\mathrm{s},\mathrm{p}}=(\pi)^p (\pi)^p = V/2^{2p}$.
    \end{enumerate}

    Using only one layer, thus taking $p=1$, we see that $V_{\mathrm{s},\mathrm{p}}=V/4$ and $V_\mathrm{s} = V/2$, thus recovering the reduction to $1/4$ and $1/2$ of the parameter space highlighted by the dotted lines in the plots 1B and 2B in Figure \ref{fig:qaoa-energylandscape}. We see that the advantage increases exponentially with the number of layers $p$.

    The parameter space can thus be restricted depending on the chosen optimization method, which usually, in the context of VQCs, is some alteration of the gradient descent. Because of these symmetries, the parameters can be initialized randomly in the interval $(\vec{\beta}, \vec{\gamma}) \in [0,\pi]^p \times[-\pi, \pi]^p$ or $(\vec{\beta}, \vec{\gamma}) \in [0,\pi]^p \times[0, \pi]^p$, depending on the presence of the linear term in the Hamiltonian. 
    
    We can also enforce the periodicity by applying two arbitrary functions whose image is the desired interval. As an example, we can choose $\sigma_\pi : \mathbb{R} \to [0, \pi]$ and $\sigma_{2\pi} : \mathbb{R} \to [-\pi, \pi]$ such that
    \begin{equation}\label{eq:scaling-func}
    \begin{aligned}
         \sigma_\pi(x) = \frac{\pi}{2}(\tanh(x)+1) \;\;\;\text{and}\;\;\; \sigma_{2\pi}(x) = \pi\tanh(x)
         .
    \end{aligned}
    \end{equation}
    The scaling function $\sigma_\pi$ will always be applied to the parameters $\vec{\beta}$, while $\sigma_\pi$ or $\sigma_{2\pi}$ will be applied to the parameters $\vec{\gamma}$ depending on the presence of the linear term in the final Hamiltonian $\mathcal{H}_\text{f}$.

    With a relatively low number of QAOA layers these scaling functions are not expected to drastically improve the convergence, but indeed with a high number of layers this is expected to be the case, considering the exponential reduction of the search space shown for by $V_\mathrm{s}$ and $V_{\mathrm{s},\mathrm{p}}$. 
    
\section{Implementation and Results}

    In this section, we present the algorithm implementation and the results of the Max Cut and the Knapsack problems. 
    
    \subsection{The Algorithm}
    We summarize the QAOA in Algorithm \ref{alg:qaoa}.

    \begin{algorithm}[H]
    \caption{Quantum Approximate Optimization Algorithm}\label{alg:qaoa}
    \begin{algorithmic}
    \Require
    \State Final Hamiltonian $\mathcal{H}_\text{f}$.
    \State Number of layers $p$.
    \State Number of shots $N_s$.
    \State Initial weights $\theta^{0} = \beta^0_1,\gamma^0_1, \dots, \beta^0_{p}, \gamma^0_{p}$.
    
    \While{the convergence of $\braket{\mathcal{H}_\text{f}} = \bra{\psi(\theta)} \mathcal{H}_\text{f} \ket{\psi(\theta)}$ is not achieved}
      
        \State Place the first layer of $H^{\bigotimes n}$ gates.
        \State Place the layers $L_1(\beta_1, \gamma_1), \dots, L_{p}(\beta_{p}, \gamma_{p})$
        \State Run the circuit $N_s$ times and sample the final state $\psi(\theta)$.
        \State Compute $\braket{\mathcal{H}_\text{f}}$.
        \State Update the rotation parameters $\theta$ with $\braket{\mathcal{H}_\text{f}}$ as cost function.
    \EndWhile
    \State Run the circuit a statistically significant number of times.
    \State Extract the state with highest probability, which is the best estimate of the optimal solution.
    \end{algorithmic}
    \end{algorithm}

    We stress the fact that, in principle, the order of $U_{\mathrm{i}}(\beta)$ and $U_{\mathrm{f}}(\gamma)$ does not matter, since we can apply the first order Trotterization of Equation (\ref{eq:final-drop-of-limit}) to $\mathcal{H}(t) =  (1-t)\mathcal{H}_\text{i} + t \mathcal{H}_\text{f}$, which yields
    \begin{equation}
       	\begin{aligned}
       		\lim\limits_{p \to +\infty} 
       		\mathcal{T}
       		\prod\limits
       		_{k=1}^{p} e^{- \mathrm{i} (1-t_k) \Delta t \mathcal{H}_\text{i} - \mathrm{i} t_k \Delta t  \mathcal{H}_\text{f}}
       		&\approx
       		\mathcal{T}
       		\prod_{k=1}^{p} 
       		e^{-\mathrm{i} (1-t_k) \Delta t \mathcal{H}_\mathrm{i}  }
       		e^{- \mathrm{i} t_k \Delta t \mathcal{H}_\mathrm{f} }
       		,
       	\end{aligned}
    \end{equation}
    or to $\mathcal{H}(t) =  t \mathcal{H}_\text{f} + (1-t)\mathcal{H}_\text{i}$, which, instead, leads to
    \begin{equation}
    	\begin{aligned}
    		\lim\limits_{p \to +\infty} 
    		\mathcal{T}
    		\prod\limits
    		_{k=1}^{p} e^{- \mathrm{i} t_k \Delta t  \mathcal{H}_\text{f}- \mathrm{i} (1-t_k) \Delta t \mathcal{H}_\text{i}}
    		&\approx
    		\mathcal{T}
    		\prod_{k=1}^{p} 
    		e^{- \mathrm{i} t_k \Delta t \mathcal{H}_\mathrm{f} }
    		e^{-\mathrm{i} (1-t_k) \Delta t \mathcal{H}_\mathrm{i}  }
    		.
    	\end{aligned}
    \end{equation}
     Introducing the free parameters $\beta_k/2, \gamma_k/2$, in the first case we obatin $L_k(\beta_k, \gamma_k) = U_{\mathrm{i}}(\beta_k)U_{\mathrm{f}}(\gamma_k)$ while in the second $L_k(\beta_k, \gamma_k) = U_{\mathrm{f}}(\gamma_k)U_{\mathrm{i}}(\beta_k)$. This does not imply that $U_\mathrm{i}$ and $U_\mathrm{f}$ commute in general, but only that there is no a priori reason to choose one ordering of $U_\mathrm{i}$ and $U_\mathrm{f}$ over the other. However, in practice, it could happen that the order affects the optimization performance.

    The optimization step is performed through an optimizer, which usually, but not necessarily, represents some variation of the gradient descent. In particular, for a quantum variational circuit, the gradient can be updated exploiting the parameter-shift rule shown in Equation (\ref{eq:parameter-shift-rule}).

    \subsection{Implementation and Results}

    We implement one instance of both the Max Cut and the Knapsack problem and apply the QAOA to find their optimal solution. The Max Cut problem is implemented as a randomly connected graph with five vertices $V = \{0, 1, 2, 3, 4\}$ and edges $E = \{(0,1), (0,2), (0,4), (1,2), (1,3), (2,3), (2,4), (3,4)\}$. The QAOA layer is repeated ten times. The Knapsack problem is implemented with five objects with randomly chosen values $V = \{4, 4, 2, 2, 4\}$ and weights $W = \{4, 3, 1, 2, 1\}$. The QAOA layer is repeated 15 times. 

    The results are shown in Figure \ref{fig:results}. In the left column we show the convergence of the cost function $\braket{\mathcal{H}_\text{f}}(\beta, \gamma)$ while on the right the distribution sampled from the output state, where the binary strings have been sorted in increasing order of the energy. The first and second rows refer to, respectively, the Max Cut and Knapsack problem. For both problems the QAOA algorithm has been run ten times with ten different sets of initial angles sampled uniformly from the interval, respectively, $[0, \pi]^5 \times [0, \pi]^5$ and $[0, \pi]^5 \times [-\pi, \pi]^5$, as discussed in Section \ref{sec:param-space-restriction}. In the cost function plot the confidence interval of one standard deviation is shown and the histogram is averaged among the ten different runs. The optimizer used is the Simultaneous Perturbation Stochastic Approximation (SPSA) optimizer with 2000 optimization steps.

    \begin{figure*}[ht]
        \centering
        \includegraphics[width = 460pt]{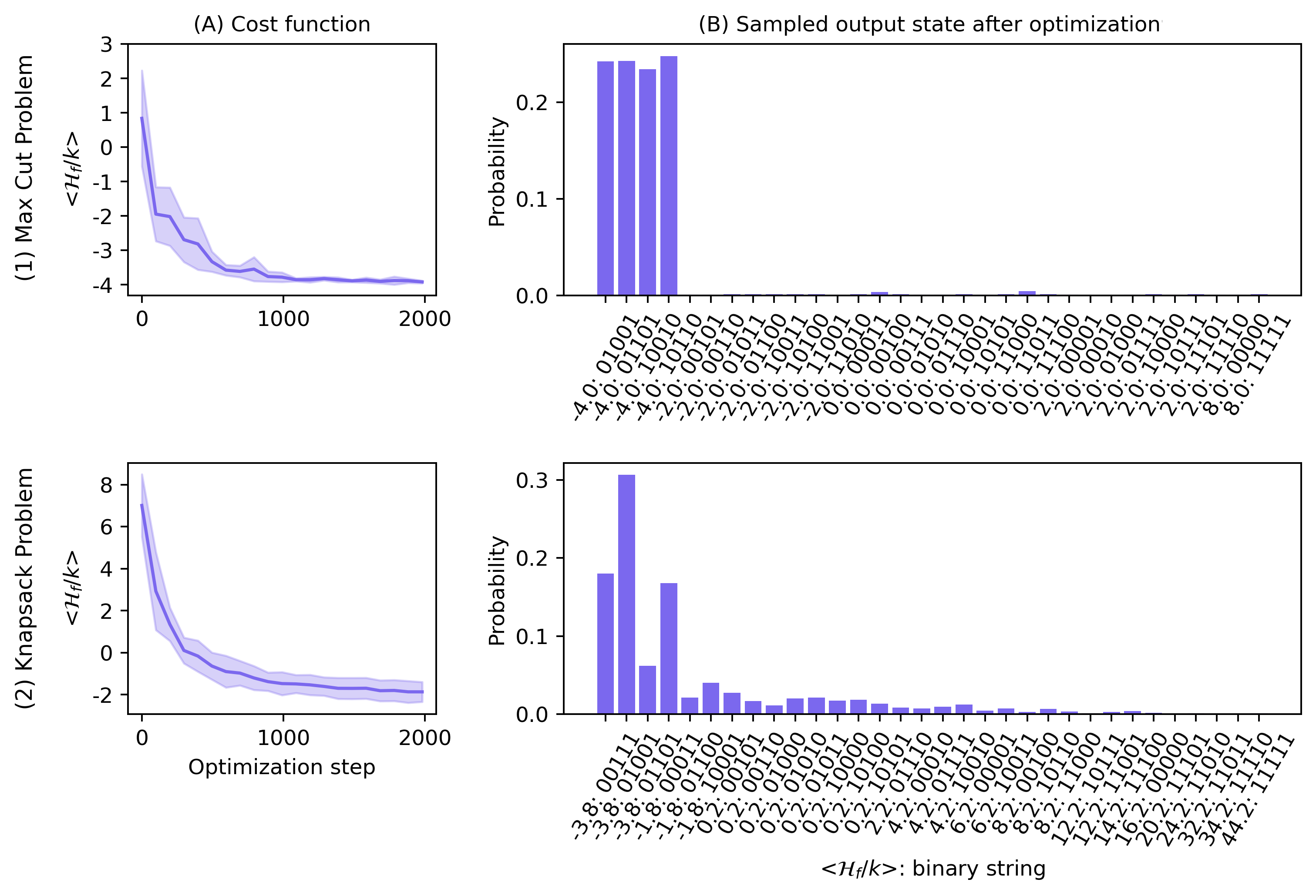}
        \vspace{-10pt}
        \caption{Plots of the results of the optimization. Row 1 and row 2: plots of the results of the QAOA circuit for, respectively, the Max Cut and the Knapsack problem run ten times each with different sets of initial angles. Column A: plots of the cost function $\braket{\mathcal{H}_\text{f}}$ decreasing during the optimization of the QAOA circuit. The standard deviation across  the ten initializations is represented by the shaded area. Column B: histogram of the sampled output state, where each binary string is associated with the scaled Hamiltonian evaluated on that string. The binary strings are organized from the lowest energy (left) to the highest energy (right). The histogram takes into account the same ten runs.}
        \label{fig:results}
    \end{figure*}

    \definecolor{light-gray}{gray}{0.95}
    \newcommand{\code}[1]{\colorbox{light-gray}{\texttt{#1}}}

\section{Polynomial Unconstrained Binary Optimization Problems}
    In this section, we extend the formulation of QUBO problems to the more general class of Polynomial Unconstrained Binary Optimization (PUBO) problems. We also outline how the QAOA can be generalized to Hamiltonians that include polynomial interaction terms.

    \subsection{PUBO Hamiltonian}

    We write the $d-$th degree cost function $C(x)$ with binary variables $x \in \{0,1\}^n$ as
    \begin{equation}
        C(x) 
        =
        \sum\limits_{k=1}^{d} 
        \sum\limits_{i_1 \ldots i_k=1}^n
        q_{k, i_1 \ldots i_k}
        x_{i_1} \ldots x_{i_k}
        ,
    \end{equation}
    where we denoted by $q_k$ the vector containing the coefficients of the $k$-th term in the sum, and by $q_{k, i_1 \ldots i_k}$ its components. 
    Following the same procedure as in Equation~(\ref{eq:qubo-ham}), we can express this cost function in terms of the spin variables $s \in \{-1, 1\}^n$ as $\mathcal{H}(s) = C\left(\frac{s+1}{2}\right)$. The change of variables yields
    \begin{equation}\label{eq:pubo-with-s+1}
    \begin{aligned}    
        \mathcal{H}(s)
        &=
        \sum\limits_{k=1}^{d} 
        \sum\limits_{i_1 \ldots i_k=1}^n
        2^{-k}
        q_{k,i_1 \ldots i_k}
        (s_{i_1} + 1) \ldots (s_{i_k} + 1)
        .
    \end{aligned}
    \end{equation}
    The coefficient tensors $q_1, \dots, q_d$ can always be symmetrized with respect to their indices, ensuring invariance under any permutation of those indices. 
    This property generalizes the symmetry of the matrix $Q$ in QUBO formulations. 
    Therefore, in the following, we assume without loss of generality that the coefficients of $q_1, \dots, q_d$ are permutation invariant.

    In Appendix~\ref{appendix:calculations}, we expand the Hamiltonian given in Equation~(\ref{eq:pubo-with-s+1}) and collect terms according to the number of spin variables in each product. The resulting Hamiltonian can be written as
    \begin{equation}\label{eq:ham-with-f}
        \mathcal{H}(s)
        =
        \sum\limits_{k = 1}^d
        \sum\limits_{i_1 \ldots i_k=1}^n
        f_{k, i_1 \dots i_k}
        s_{i_1} \ldots s_{i_k}
        ,
    \end{equation}
    with $f_{k, i_1 \dots i_k}$ linear functions of the original coefficients $q_1, \dots, q_d$ given by
    \begin{equation}\label{eq:f-coefficients}
        f_{k, i_1 \dots i_k}
        =
        \sum\limits_{i_1 \dots i_k = 1}^n
        2^{-k}
        \left(
            q_{k, i_1 \dots i_k}
            +
            \sum\limits_{h = 1}^{d-k} 
            \sum\limits_{i_{k+1} \dots i_{k+h} = 1}^n
            \binom{k+h}{k}
            q_{k, i_1 \dots i_{k+h}}
        \right)
        .
    \end{equation}
    Since we assumed invariance under permutations, we have that $q_{k, \sigma(i_1 \ldots i_k)} = q_{k, i_1 \ldots i_k}$, where $\sigma$ is a particular permutation of the indices $i_1, \dots, i_k$. Using this property it is easy to verify in Equation (\ref{eq:f-coefficients}) that also $f_{k, \sigma(i_1 \dots i_k)} = f_{k, i_1 \dots i_k}$. We can then remove the redundant terms from the sum of Equation (\ref{eq:ham-with-f}) by defining some coefficients $a_{k, i_1 \ldots i_k}$ that count the redundancy. Namely, we define
    \begin{equation}
        a_{k, i_1 \ldots i_k}
        :=
        |\sigma(i_1, \ldots, i_k)| f_{k,i_1 \dots i_k}
        ,
    \end{equation}
    where $|\sigma(i_1, \ldots, i_k)|$ is the number of permutations of the indices $i_1, \dots, i_k$. This allows us to rewrite Equation (\ref{eq:ham-with-f}) with ordered indices as
    \begin{equation}\label{eq:ham-pubo-spin}
        \mathcal{H}(s)
        =
        \sum\limits_{k = 1}^d
        \sum\limits_{i_1 < \ldots < i_k=1}^n
        a_{k, i_1 \ldots i_k}
        s_{i_1} \ldots s_{i_k}
        ,
    \end{equation}
    which is the generalization of Equation (\ref{eq:ham-cost-function}). We will set Equation (\ref{eq:ham-pubo-spin}) as the final Hamiltonian $\mathcal{H}_\text{f}$ substituting the $s$ variables with Pauli spins $\sigma_z$ operators. Namely, 
    \begin{equation}\label{eq:ham-pubo-spin-with-sigmas}
        \mathcal{H}_\text{f}(s)
        =
        \sum\limits_{k = 1}^d
        \sum\limits_{i_1 < \ldots < i_k=1}^n
        a_{k, i_1 \ldots i_k}
        \sigma_z^{i_1} \ldots \sigma_z^{i_k}
        .
    \end{equation}
    It is clear that in many cases it could be more convenient, when possible, to directly phrase the optimization cost function in terms of the spin variables $s \in \{-1,1\}^n$ rather than in terms of the binary variables $x \in \{0,1 \}^n$, so to avoid dealing with Equation (\ref{eq:f-coefficients}).

    \subsection{QAOA for PUBO Problems}
    We now show how to generalize the QAOA algorithm to PUBO problems. The framework developed so far will remain the same: we will implement the Trotterization of the exponential of the evolving Hamiltonian $\mathcal{H}(t) = (1-t)\mathcal{H}_\text{i} + t \mathcal{H}_\text{f} \; \text{ with } t \in [0,1]$ and compute the general $p-$th layer operator $L_p = e^{-\mathrm{i} \frac{\beta_p}{2} \mathcal{H}_\text{i}}e^{-\mathrm{i}\frac{\gamma_p}{2} \mathcal{H}_\text{f}} = U_{\mathrm{i}}(\beta_p) U_{\mathrm{f}}(\gamma_p)$. We use the same initial Hamiltonian, which, as shown in Equation (\ref{eq:exp-of-sigmax}), leads to 
    \begin{equation}
        U_{\mathrm{i}}(\beta_p) =\prod\limits_{i = 1}^n R^i_x(\beta_p).
    \end{equation}
    The only difference is the term $U_{\mathrm{f}}(\gamma_p)$, which is given by 
    \begin{equation}
        U_{\mathrm{f}}(\gamma_p) = e^{-\mathrm{i}\frac{\gamma_p}{2} \mathcal{H}_\text{f}},
    \end{equation}
    with $\mathcal{H}_\text{f}(s)$ as defined in Equation (\ref{eq:ham-pubo-spin-with-sigmas}). We then compute 
    \begin{equation}\label{eq:U-f-PUBO}
        \begin{aligned}
            U_{\mathrm{f}}(\gamma_p) 
            &=
            e^{-\mathrm{i} \frac{\gamma_p}{2} \mathcal{H}_\text{f}} 
            \\
            &=
            \exp \left\{
            -\mathrm{i} \sum\limits_{k = 1}^d
            \sum\limits_{i_1 < \ldots < i_k=1}^n
            \frac{\gamma_p a_{k, i_1 \ldots i_k}}{2}
            \sigma_z^{i_1} \ldots \sigma_z^{i_k}
            \right\} 
            \\
            &=
            \prod\limits_{\substack{k = 1}}^d
            \prod\limits_{i_1 < \ldots < i_k=1}^n
            \exp\left\{-\mathrm{i} \frac{\gamma_p a_{k, i_1 \ldots i_k}}{2} \sigma_z^{i_1} \ldots \sigma_z^{i_k} \right\} 
            ,
        \end{aligned} 
        \end{equation}
        where in the last equality we separated the exponential of the sums in the product of the exponential since $[\sigma_z^{i}, \sigma_z^{j}]=0$ $\forall 
        i,j$. Note that, a priori, the order of the exponential functions does not matter since they all commute.

        We generalize the $R_{ZZ}(\theta)$ gate of Equation (\ref{eq:Rzz}) by defining 
        \begin{equation}\label{eq:R-Z-k}
            R_{Z^k}(\theta) := \exp\left\{-\mathrm{i} \frac{\theta}{2} \sigma_z^{i_1} \ldots \sigma_z^{i_k} \right\}
            .
        \end{equation}
        In Appendix \ref{appendix:pubo-qaoa} we compute the explicit form of $R_{Z^k}$, which is given by
        \begin{equation}\label{eq:R-Z-k-decomposed.}
            R_{Z^k}(\theta) = 
            \text{CNOT}(i_1, i_k) \dots \text{CNOT}(i_{k-1}, i_k)
            \cdot
            \mathbb{I} \otimes \ldots \otimes \mathbb{I} \otimes R_Z(\theta)
            \cdot
            \text{CNOT}(i_{k-1}, i_k) \dots \text{CNOT}(i_1, i_k)
            .
        \end{equation}
        While in Appendix \ref{appendix:pubo-qaoa} we performed explicitly the (admittedly tedious) calculations that lead to Equation (\ref{eq:R-Z-k-decomposed.}), we mention that an alternative and intuitive derivation of the same result is discussed in \cite{nielsen-chuang}, although it differs in that it makes use of an ancilla qubit. Here, the authors note that $R_{Z^k}$ applies a phase shift of $e^{-\mathrm{i}\theta/2}$ or $e^{\mathrm{i}\theta/2}$ depending on the total parity $\pm1$ given by the product $\sigma_z^{i_1}\ldots\sigma_z^{i_k}$. To implement the gate, we can then compute the total parity via CNOTs, store it in the ancilla qubit indexed by $k+1$, apply the conditional phase $\exp\left\{-\mathrm{i}(\theta/2)\sigma_z^{i_{k+1}}\right\}$ on the ancilla qubit, and then uncompute its parity to restore its initial value.
        
        \subsection{The Circuit and its Symmetries}

        The generic $p-$th layer operator $L_p(\beta_p, \gamma_p)$ of the QAOA algorithm for a PUBO problem is
        \begin{equation}\label{eq:layer-PUBO}
            L_p(\beta_p, \gamma_p) 
            = 
            U_{\mathrm{i}}(\beta_p) U_{\mathrm{f}}(\gamma_p)
            =
            \left(
                \prod\limits_{i = 1}^{n}
                R^i_x(\beta_p)
            \right)
            \left(
                \prod\limits_{k=1}^d
                \prod\limits_{i_1 < \ldots < i_k = 1 }^{n} 
                R_{Z^k}(\gamma_p a_{k, i_1, \dots, i_k})
            \right)
            .
        \end{equation}
        While the $U_{\mathrm{i}}$ operator stays the same regardless of the degree of the final Hamiltonian $\mathcal{H}_{\text{f}}$, $U_{\mathrm{f}}$ depends on it. For each term of the Hamiltonian with spins $k=1, \dots, d$ we associate a $R_{Z^k}$ gate, which is in turn decomposed into a series of CNOT gates and one single $R_Z$ gate as shown in Equation (\ref{eq:R-Z-k-decomposed.}).
        
        The complete circuit that implements the QAOA for PUBO problems has the same architecture represented in Figures \ref{fig:qaoa-circuit} and \ref{fig:KSP-circuit} but with a different $U_\mathrm{f}$ operator, which is the one shown in Figure \ref{fig:Uf-PUBO-decomposition}.

        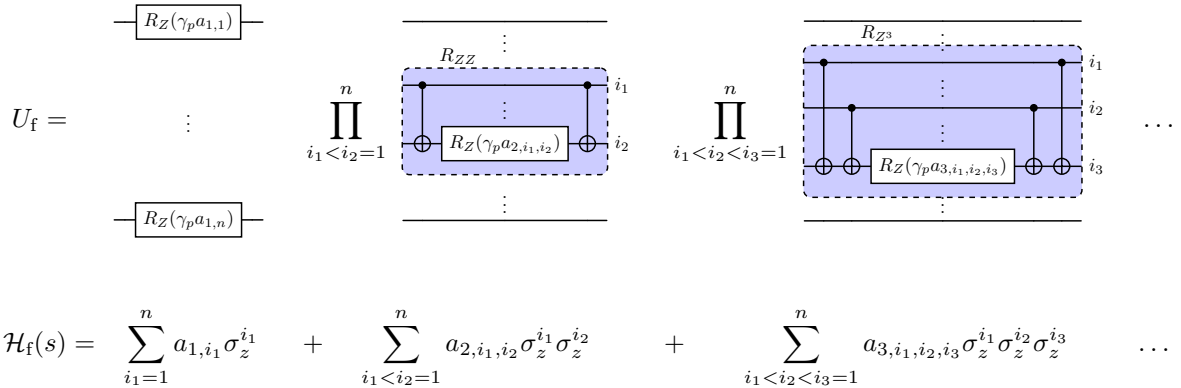
\begin{figure}[h!]

        $$
        \newcommand{\CircuitDots}{\gate[wires=1,style={fill=none,draw=none,text height=1cm}]{\substack{\vdots \\[0.1cm]}}}
        \begin{aligned}
            U_\mathrm{f}
            &=
            \hspace{10pt}
            \begin{tikzpicture}[baseline={([yshift=-.5ex]current bounding box.center)},vertex/.style={anchor=base,
            circle,fill=black!25,minimum size=10pt,inner sep=0pt}]
            \node[scale=0.7] {
                \begin{quantikz}[row sep=-0.1cm, column sep=0.2cm]
                                &  & \gate{R_Z(\gamma_p a_{1, 1})} &  & \\[34pt]
                \setwiretype{n} &  & \CircuitDots         &  & \\[34pt]
                                &  & \gate{R_Z(\gamma_p a_{1, n})} &  &
                 \end{quantikz}
                };
            \end{tikzpicture}
            \hspace{10pt}
            \prod\limits_{i_1 < i_2 = 1}^n
            \begin{tikzpicture}[baseline={([yshift=-.5ex]current bounding box.center)},vertex/.style={anchor=base,
            circle,fill=black!25,minimum size=18pt,inner sep=2pt}]
            \node[scale=0.7] {
                \begin{quantikz}[row sep=-0.1cm, column sep=0.2cm]
                                &        &                           &          & \\
                \setwiretype{n} &          & \CircuitDots       &          & \\[12pt]
                                & \ctrl{2}\gategroup[3,steps=3,style={dashed,rounded corners, fill=blue!20, inner xsep=+3pt}, background, label style={xshift=-25pt}]{{\sc $R_{ZZ}$}} &                           & \ctrl{2} & \rstick{$i_1$} \\
                \setwiretype{n} &          & \CircuitDots              &          & \\
                                & \targ{}  & \gate{R_Z(\gamma_p a_{2,i_1,i_2})} & \targ{}  & \rstick{$i_2$}  \\[12pt]
                \setwiretype{n} &          & \CircuitDots              &          & \\
                                &          &                           &          & 
                \end{quantikz}
                };
            \end{tikzpicture}
            \hspace{10pt}
            \prod\limits_{i_1 < i_2 < i_3 = 1}^n
            \begin{tikzpicture}[baseline={([yshift=-.5ex]current bounding box.center)},vertex/.style={anchor=base,
            circle,fill=black!25,minimum size=18pt,inner sep=2pt}]
            \node[scale=0.7] {
                \begin{quantikz}[row sep=-0.1cm, column sep=0.2cm]
                                &          &          &                               &          &          & \\
                \setwiretype{n} &          &          & \CircuitDots                  &          &          & \\
                                & \ctrl{4}\gategroup[5,steps=5,style={dashed,rounded corners, fill=blue!20, inner xsep=+3pt}, background, label style={xshift=-35pt}]{{\sc $R_{Z^3}$}} &          &                               &          & \ctrl{4} & \rstick{$i_1$} \\
                \setwiretype{n} &          &          & \CircuitDots                  &          &          & \\
                                &          & \ctrl{2} &                               & \ctrl{2} &          & \rstick{$i_2$} \\
                \setwiretype{n} &          &          & \CircuitDots                  &          &          & \\
                                & \targ{}  & \targ{}  & \gate{R_Z(\gamma_p a_{3,i_1,i_2,i_3})} & \targ{}  & \targ{}  & \rstick{$i_3$} \\
                \setwiretype{n} &          &          & \CircuitDots                  &          &          & \\
                                &          &          &                               &          &          & 
                \end{quantikz}
                };
            \end{tikzpicture}
            & \dots
        \end{aligned}
        $$
        
        $$
        \newcommand{\CircuitDots}{\gate[wires=1,style={fill=none,draw=none,text height=1cm}]{\substack{\vdots \\[0.1cm]}}}
        \begin{aligned}
            \hspace{-5pt}
            \mathcal{H}_\text{f}(s)
            &=
            \hspace{8pt}
            \sum\limits_{i_1=1}^n
            a_{1, i_1}
            \sigma^{i_1}_z
            \hspace{14pt}
            +
            \hspace{12pt}
            \sum\limits_{i_1 < i_2=1}^n    
            a_{2, i_1, i_2}
            \sigma^{i_1}_z\sigma^{i_2}_z
            \hspace{25pt}
            +
            \hspace{20pt}
            \sum\limits_{i_1 < i_2 < i_3 = 1}^n    
            a_{3, i_1, i_2, i_3}
            \sigma^{i_1}_z\sigma^{i_2}_z\sigma^{i_3}_z
            \hspace{16pt}
            &\hspace{0pt} \dots
        \end{aligned}
        $$
        \caption{Quantum gate decomposition of the $U_\mathrm{f}$ operator. To every term with $k$ interactions of the spins in the final Hamiltonian corresponds one $R_{Z^k}$ gate. Note that, since the $R_{Z^k}$ gates all commute, their order must not necessarily be the one shown here.}
        \label{fig:Uf-PUBO-decomposition}
        \end{figure}
    Once again, the parameters $\beta_p, \gamma_p$ are symmetric with respect to the origin. Namely, we can claim the following proposition.

    \begin{proposition}[QAOA Symmetry for PUBO problems]\label{prop:simmetry-pubo}
        The energy landscape of the QAOA that implements a PUBO problem is symmetric with respect to the origin. Namely 
        \begin{equation}\label{eq:simmetry-pubo}
            \braket{\mathcal{H}_\mathrm{f}}(-\vec{\beta}, -\vec{\gamma}) = \braket{\mathcal{H}_\mathrm{f}}( \vec{\beta}, \vec{\gamma})
            .
        \end{equation}
    \end{proposition}

    Proposition \ref{prop:simmetry-pubo} can be proven in the same way as in Proposition \ref{prop:simmetry}, since in the previous proof we did not assume anything about the form of $\mathcal{H}_{\text{f}}$, except for the fact that it is real. The results thus also apply for the PUBO Hamiltonian and more in general for real target Hamiltonians.

    Analogously to Proposition \ref{prop:period}, we can in this case claim the following result.

    \begin{proposition}[QAOA Periodicity for PUBO problems]\label{prop:period-pubo}
        If the final Hamiltonian $\mathcal{H}_\mathrm{f}$ of the QAOA has only terms with an even number of interactions in the spins then its associated energy landscape is periodic in its variables $\vec{\beta}$ with period $\pi$. Namely,
        \begin{equation}\label{eq:main-claim-pubo}
            \braket{\mathcal{H}_\mathrm{f}}(\beta_1, \dots, \beta_i + \pi, \dots \beta_{p}, \vec{\gamma}) = \braket{\mathcal{H}_\mathrm{f}}(\vec{\beta}, \vec{\gamma})
            \;\;\;
            \forall i
            .
        \end{equation}
    \end{proposition}

    \begin{proof}
        We start by generalizing the calculation in Equation (\ref{eq:commutator-proof}). We have that
        \begin{equation}
            \sigma_x^{\otimes k} \cdot \sigma_z^{\otimes k} = (\sigma_x \sigma_z)^{\otimes k} = (-\sigma_z \sigma_x)^{\otimes k} = (-1)^{k} \sigma_z^{\otimes k} \cdot \sigma_x^{\otimes k} 
            ,
        \end{equation}
       which shows that, if $k$ is even, then $[\sigma_x^{\otimes k}, \sigma_z^{\otimes k}] = 0$. As discussed in the proof of Proposition \ref{lemma:commutation}, this implies that 
       \begin{equation}\label{eq:commutator-sigmax-Rz^k}
            \left[
                \sigma_x^{\otimes k}, 
                R_{Z^k}(\theta)
            \right] = 0
            ,
        \end{equation}
        and consequently that 
        \begin{equation}
            \left[\sigma^1_x \dots \sigma^n_x, L_k(\beta_k, \gamma_k)\right] = 0
            .
        \end{equation}
        The rest of the proof in Proposition \ref{prop:period} can be followed analogously to obtain the same result. 

      \end{proof}

	\section{Further Reading}
	Up to this point, we have introduced the foundations of QAOA and examined several of its mathematical properties. Nonetheless, important questions remain: How significant is the error from the Trotter approximation? How tractable is the optimization landscape? And can variations of QAOA improve performance? In this section, we briefly outline these issues and highlight relevant material in the literature.
	
	\subsection*{Trotter Error}
	In the field of quantum simulation, Trotterization is one of the most widely used techniques for approximating time evolution, and its accuracy has been analyzed across many settings \cite{1912.08854, PhysRevLett.128.210501, 2405.01131, PhysRevResearch.6.033285}. When introducing the Trotter approximation shown in Proposition \ref{prop:trotter}, it is natural to ask how well this approximation performs when dealing with a Hamiltonian of interest.
	We know that increasing the number of Trotter layers systematically reduces the approximation error, and in the limit of infinitely many layers the sequence of layers converges to the exact unitary time evolution. How big is the error when truncating the limit?
	Suppose we have a Hamiltonian $\mathcal{H} = \sum\limits_{k = 1}^m \mathcal{H}_k$ which generates an exact time evolution $U(t) =e^{- \mathrm{i} t \mathcal{H}_k}$ and its Trotterized time evolution $T_1(t) := \prod\limits_{k = 1}^m e^{- \mathrm{i} t \mathcal{H}_k}$ and define the Trotter error as the norm of the difference between the approximated evolution and the exact one, namely $E:= \|T_1(t) -  e^{- \mathrm{i} t \mathcal{H}}\|$, where the norm $\| \cdot \|$ that is usually chosen is the spectral norm.
	In \cite{1912.08854} it has been shown that this error scales quadratically in time, as we expect from Proposition  \ref{prop:trotter}, and is determined by the commutator structure of the Hamiltonian. In particular, the following bound has been proven:
	\begin{equation}
		\bigl\| T_1(t)-U(t)\bigr\|
		\;\le\;
		\frac{t^{2}}{2}
		\sum_{k_1=1}^{m}
		\left\|
		\left[
		\sum_{k_2=k_1+1}^{m}
		H_{k_2},
		\,H_{k_1}
		\right]
		\right\|
		,
	\end{equation}
	which shows that the error, when applying the first-order Trotterization, grows like $t^{2}$ and is governed entirely by how strongly the Hamiltonian terms fail to commute. If many terms commute than the Trotter error becomes smaller.
	
	\subsection*{Expressivity and Barren Plateaus}
	Another important consideration when working with VQCs, such as those used in the QAOA algorithm, is that gradient-based optimization methods can suffer from the \textit{barren plateau} phenomenon. This phenomenon, which has been extensively studied \cite{PRXQuantum.3.010313,Ragone2024, 2307.08956,PRXQuantum.2.040316,1903.05076}, can severely hinder the optimization of VQCs by causing the gradients to vanish exponentially with system size. More precisely, consider a VQC with parameters $\vec{\theta}$ and a cost function $C(\vec{\theta}) = \operatorname{Tr}\left[ \rho(\vec{\theta}) \mathcal{H} \right]$,
	where $\rho(\vec{\theta}) = U(\vec{\theta})\rho U^\dagger(\vec{\theta})$ is the state prepared by the circuit and $\mathcal{H}$ is some Hermitian operator.
	The gradient with respect to a single parameter $\theta_i$ is $\partial_{\theta_i}C
	:= \frac{\partial C}{\partial \theta_i}  \operatorname{Tr}\left[ \rho(\vec{\theta}) \mathcal{H} \right].$
	It can be shown that, for sufficiently deep or highly expressive circuits, the expectation of this gradient over random parameter initializations satisfies
	\begin{equation}
		\mathbb{E}\left[\partial_{\theta_i}C\right] = 0.
	\end{equation}
	The barren plateau phenomenon consists in the fact that the variance of this gradient, for a circuit whose architecture is capable of approximating a wide variety of unitary matrices, vanishes exponentially with the number of qubits $n$:
	\begin{equation}
		\mathrm{Var}\left[\partial_{\theta_i}C\right] \sim \mathcal{O}\left( \frac{1}{2^n} \right).
	\end{equation}
	As a result, combining the facts that the gradient has zero mean and a variance that exponentially goes to zero, this implies that the gradient quickly becomes extremely small when increasing the number of qubits, making gradient-based optimization infeasible for large systems. As mentioned, the results hold on average and for architectures that are highly expressive, which means that they can approximate many unitary matrices. More precisely, this results hold for 2-designs, which are defined as a set of unitary operators $\{U_k\}$ acting on an $n$ qubits such that, for any operator $O$, the average over the set reproduces the Haar-random second moment:
	\begin{equation}
		\frac{1}{|\{U_k\}|} \sum_k U_k O U_k^\dagger \otimes U_k O U_k^\dagger
		=
		\int dU_{\mathrm{Haar}} \, U O U^\dagger \otimes U O U^\dagger,
	\end{equation}
	where the Haar measure is the group invariant measure, which is the equivalent of a uniform distribution on groups, and the integral is carried out on the group of unitary matrices. The Haar measure finds many applications in quantum information theory \cite{2307.08956} since, playing the role of a uniform distribution, it can be used to estimate expectation values on the circuit architectures.
	
	\subsection*{Variants of the QAOA}
	A vast ecosystem of QAOA variations has been proposed \cite{2306.09198} to overcome its limitations, each introducing distinct modifications to the ansatz structure or the parameter optimization strategy. Here we briefly mention three variations to the QAOA algorithm and highlight their main differences and advantages.
	
	A straightforward yet powerful enhancement is the \textit{Multi-Angle QAOA} (ma-QAOA). While the standard QAOA uses a single parameter $\beta$ for the unitary $U_\mathrm{i}$ and $\gamma$ for the unitary $U_\mathrm{f}$ at each step $p$, ma-QAOA introduces a unique parameter for each individual term in the Hamiltonians. Formally, the unitaries become:
	\begin{equation}
		U_{\mathrm{i}}(\vec{\beta}_p)
		=
		\prod\limits_{i = 1}^n e^{-\mathrm{i} \frac{\beta_{p,i}}{2} \sigma^i_x}
		\hspace{15pt}
		\text{and} 
		\hspace{15pt}
		U_{\mathrm{f}}(\vec{\gamma}_p)
		=
		\prod\limits_{\substack{i,j = 1 \\ i\leq j}}^n
		\exp\left\{-\mathrm{i} \frac{\gamma_{p,ij}a^{ij}}{2} \sigma^i_z \sigma^j_z \right\} 
		\prod\limits_{i= 1}^n R^i_z\left(\gamma_{p,i} b^{i}\right) 
	\end{equation}
	This increases the number of classical parameters and improves upon QAOA by providing a much more expressive architecture, allowing the algorithm to fine-tune each $R_{Z^k}$ gate independently.
	
	Another approach, inspired by methods to speed up quantum processes, is called \textit{Digitized Counterdiabatic QAOA} (DC-QAOA). The standard QAOA tries to mimic a slow, smooth transition to the solution, which often requires a very deep circuit (a high value of $p$) to work well. DC-QAOA adds a ``shortcut" to this process by introducing an extra component, known as a counterdiabatic (CD) term. This modifies the standard gate $L(\beta, \gamma)$ to $L(\beta, \gamma,\alpha)$, where $\alpha$ controls an additional quantum operator
	\begin{equation}
		U_{\mathrm{CD}}(\alpha)=\prod_{j=1}^{L}\exp(-i\alpha A^{q}_{t}),
	\end{equation}
	where \(A_t^{q}\) is a \(q\)-local operator drawn from a commutator-based operator basis constructed from \(\mathcal{H}_{\mathrm{i}}\) and \(\mathcal{H}_{\mathrm{f}}\). The commutators capture the leading directions along which non-adiabatic transitions occur due to the non-commutativity of the Hamiltonians. Including the unitary \(e^{-i\alpha A_t^q}\) rotates the state along these directions, effectively canceling diabatic leakage and allowing QAOA to stay closer to the instantaneous ground state, achieving high-quality solutions with smaller depth \(p\).
	
	Finally, \textit{Warm-Starting QAOA} (WS-QAOA), tackles the problem of poor initialization. Instead of starting from the uniform superposition $\left|+\right\rangle^{\otimes n}$, it initializes the state using a classical solution to a relaxed version of the problem. For a solution $c_i^* \in [0,1]$ from a continuous relaxation, the initial state is
	\begin{equation}
		\ket{s^{*}}=\bigotimes_{i=1}^{n}R_{y}(\theta_{i})\ket{0}^{\otimes n} \quad \text{where} \quad \theta_{i}=2\arcsin\left(\sqrt{c_{i}^{*}}\,\right).
	\end{equation}
	The initial Hamiltonian is then modified to $\mathcal{H}_{\mathrm{i}}^{(\text{ws})}=-\sin(\theta_{i})X-\cos(\theta_{i})Z$ to preserve this initial state. This should improve QAOA, especially at low depths $p$, by providing a high-quality starting point: it guarantees that the quantum solution is at least as good as the best classical approximation from the start, effectively giving QAOA a ``head start''.

    \appendix
    \section{Appendix}
    
    \subsection{Approximation of the Time-Dependent Hamiltonian Evolution}\label{app:approximation-of-dyson}

    In this section, we take a closer look at the approximation performed in Equation (\ref{eq:approximation-of-U}). We begin by deriving the solution to the Schrödinger equation with a time-dependent Hamiltonian and we proceed discussing how the time evolution operator can be approximated. In the following we set $\hbar = 1$.
    
    \subsection*{Dyson Series}
    The dynamic of a pure quantum state is given by the Schrödinger equation
    \begin{equation}\label{eq:Schr-eq}
    	\frac{d}{d t} \psi(t) = - \mathrm{i} \mathcal{H}(t)\psi(t)
    	.
    \end{equation}
    With the ansatz $\psi(t) = U(t) \psi_0$ we can rewrite Equation (\ref{eq:Schr-eq}) as 
    \begin{equation}\label{eq:schr-U}
    	\frac{d}{d t} U(t)= - \mathrm{i} \mathcal{H}(t)U(t)
    	.
    \end{equation}
    If $\mathcal{H}(t) = \mathcal{H}$ has no explicit dependence on time, then the solution to $\frac{d}{dt} U(t) = -\mathrm{i} \mathcal{H} U(t)$ is given by $U(t) = \exp\{-it\mathcal{H}\}$, which can be easily expanded in Taylor series and derived to confirm it. If, instead, $\mathcal{H}(t)$ has an explicit dependence on time, then $U(t) = \exp\{-it\mathcal{H}(t)\}$ is not a valid solution. In fact by deriving it we obtain
    \begin{equation}
    	\frac{d}{dt}U(t) = \frac{d}{dt} \sum\limits_{k = 0}^\infty \frac{1}{k!}(-\mathrm{i}t\mathcal{H}(t))^k = 
    	-\mathrm{i}\mathcal{H}(t)U(t)-\mathrm{i}tU(t)\frac{d}{dt}\mathcal{H}(t) \neq- \mathrm{i}\mathcal{H}(t)U(t)
    \end{equation}
    To obtain a solution for Equation (\ref{eq:schr-U}) we start by rewriting it as an integral equation: we integrate on both sides and, considering the initial condition $U(t_0) = \mathbb{I}$, we get to the form
    \begin{equation}
    	U(t) = \mathbb{I} - \mathrm{i} \int_{t_0}^t ds \mathcal{H}(s) U(s)
    	.
    \end{equation}
    We can then iteratively plug $U(t)$ so defined in the right-hand side of the integral equation and thus obtain
    \begin{align}
    	U_1(t) &= \mathbb{I} - \mathrm{i} \int_{t_0}^t dt_1 \, \mathcal{H}(t_1), \\[1mm]
    	U_2(t) &= \mathbb{I} - \mathrm{i} \int_{t_0}^t dt_2 \, \mathcal{H}(t_2) 
    	\left( \mathbb{I} - \mathrm{i} \int_{t_0}^{t_2} dt_1 \, \mathcal{H}(t_1) \right) \\[1mm]
    	&= \mathbb{I} - \mathrm{i} \int_{t_0}^t dt_2 \, \mathcal{H}(t_2) 
    	+ (-\mathrm{i})^2 \int_{t_0}^t dt_2 \int_{t_0}^{t_2} dt_1 \, \mathcal{H}(t_2) \mathcal{H}(t_1), \\[1mm]
    	&\dots \\[1mm]
    	U_n(t) &= \mathbb{I} + \sum_{k=1}^n (-\mathrm{i})^k 
    	\int_{t_0}^t dt_k \int_{t_0}^{t_k} dt_{k-1} \cdots \int_{t_0}^{t_2} dt_1 \; 
    	\mathcal{H}(t_k) \mathcal{H}(t_{k-1}) \cdots \mathcal{H}(t_1)
    	.
    \end{align}   
    The operator $U(t)$ is obtained in the limit $n \to \infty$, which yields
    \begin{equation}\label{eq:U-without-T}
    	U(t) = \mathbb{I} + \sum_{k=1}^\infty (-\mathrm{i})^k 
    	\int_{t_0}^t dt_k \int_{t_0}^{t_k} dt_{k-1} \cdots \int_{t_0}^{t_2} dt_1 \; 
    	\mathcal{H}(t_k) \mathcal{H}(t_{k-1}) \cdots \mathcal{H}(t_1)
    	,
    \end{equation}
    where we stress that $t_0 \leq t_1 \leq t_2 \leq \dots\leq t_k \leq t$, so that the operators $\mathcal{H}(t_k) \dots \mathcal{H}(t_1)$ are time-ordered in descending time. Usually, the expression in Equation~(\ref{eq:U-without-T}) is manipulated in such a way that it can be rewritten as an exponential function of the Hamiltonian, so that we obtain a solution similar in notation to the time-independent case. 
    This is done by first decoupling the nested integrals so that they can be replaced by multiple integrals with integration variables that range independently from $t_0$ to $t$. In doing this, every $i-$th variable is set in the range $t_0 \leq t_i \leq t$, which implies that the operators $\mathcal{H}(t_k) \dots \mathcal{H}(t_1)$ might not be in time-descending order when evaluated at some particular combination of $t_1, \dots, t_k$. In order to decouple the integrals we then also introduce an operator, called time-ordering operator, and denoted as $\mathcal{T}$, which reorders the terms $\mathcal{H}(t_k) \dots \mathcal{H}(t_1)$ in time-descending order. Also, we notice that the volume delimited by $t_0 \leq t_1 \leq t_2 \leq \dots\leq t_k \leq t$ is a $k$-simplex (in two dimensions this reduces to a triangle) inside the hypercube of volume $(t-t_0)^k$. The volume of the $k$-simplex is $(t-t_0)^k/k!$, and if we allow integration over the whole hypercube delimited by $t_0 \leq t_i \leq t$ $\forall i$ we must divide the result by $k!$ to obtain the original volume. In conclusion we obtain an expression similar to Equation (\ref{eq:U-without-T}) but where we let every variable freely run from $t_0$ to $t$ by ensuring the right time ordering of the Hamiltonian operators and where de divided the hyper-volume by the factor $k!$. We then rewrite $U(t)$ as 
    \begin{equation}
    	U(t) = \mathbb{I} + \sum_{k=1}^\infty \frac{(-\mathrm{i})^k}{k!} 
    	\int_{t_0}^t dt_k \int_{t_0}^{t} dt_{k-1} \cdots \int_{t_0}^{t} dt_1 \; 
    	\mathcal{T} \left[\mathcal{H}(t_k) \mathcal{H}(t_{k-1}) \cdots \mathcal{H}(t_1)\right]
    	,
    \end{equation}
    which we can write compactly as 
    \begin{equation}\label{eq:time-ordered-exp}
    	U(t) = \mathcal{T} \exp\!\Bigg(-\mathrm{i} \int_{t_0}^t \mathcal{H}(s)\, ds \Bigg).
    \end{equation}
    A more in depth discussion regarding the time-ordering operator and the Dyson series can be found in \cite{dyson}.
    
    \subsection*{Approximation of the Evolution Operator}
    In this subsection we show how we can approximte the time evolution operator $U(t)$ of Equation (\ref{eq:time-ordered-exp}). We start with the following proposition. 
    \begin{proposition}[Dyson Series as a Product of Exponentials]
    	Let $\mathcal{H}(t)$ be a time-dependent Hamiltonian. Then, the time-ordered exponential can be expressed as the limit of a product of exponentials:
    	\begin{equation}\label{eq:dyson-approx}
    		\mathcal{T} \exp\!\Bigg(-\mathrm{i} \int_{t_0}^t \mathcal{H}(s)\, ds \Bigg) 
    		= 
    		\lim_{p \to +\infty} 
    		\mathcal{T}
    		\prod_{k=1}^{p} 
    		\exp\big\{- \mathrm{i} \Delta t \mathcal{H}(t_k)\big\},
    	\end{equation}
    	where $\Delta t := t/p$ and $t_k = k \Delta t$.
    \end{proposition}
    
    \begin{proof}
    	We start from the right-hand side of Equation (\ref{eq:dyson-approx}) and expand each exponential:
    	\begin{equation}
    		\begin{aligned} 
    			&\lim_{p \to +\infty}
    			\mathcal{T} 
    			\prod_{k=1}^{p} \Big( \mathbb{I} - \mathrm{i} \Delta t \mathcal{H}(t_k) + \mathcal{O}(\Delta t^2) \Big) \\
    			&= \lim_{p \to +\infty} \Bigg[ 
    			\mathbb{I} 
    			+ \sum_{n=1}^{p} (-\mathrm{i} \Delta t)^n 
    			\sum_{0 \le k_1 < \dots < k_n \le p-1} 
    			\mathcal{H}(t_{k_n}) \cdots \mathcal{H}(t_{k_1}) 
    			+ \mathcal{O}(\Delta t)
    			\Bigg] \\
    			&= \mathbb{I} + \sum_{n=1}^\infty (-\mathrm{i})^n 
    			\int_{t_0}^t dt_n \int_{t_0}^{t_n} dt_{n-1} \cdots \int_{t_0}^{t_2} dt_1 \; 
    			\mathcal{H}(t_n) \mathcal{H}(t_{n-1}) \cdots \mathcal{H}(t_1) \\
    			&= \mathcal{T} \exp\Big(-\mathrm{i} \int_{t_0}^t \mathcal{H}(t)\, dt \Big),
    		\end{aligned}
    	\end{equation}
    	which proves the proposition.
    \end{proof}
    
    We now show how we can further approximate the term 
    \begin{equation}\label{eq:single-term-trotter}
    	\exp\left\{- \mathrm{i} \Delta t \mathcal{H}(t_k)\right\} = 
    	\exp\left\{- \mathrm{i} (1-t_k) \Delta t  \mathcal{H}_\mathrm{i} - \mathrm{i} t  \Delta t  \mathcal{H}_\mathrm{f}\right\}
    	,
    \end{equation}
    where we inserted the time-evolving Hamiltonian $\mathcal{H}(t) = (1-t) \mathcal{H}_\mathrm{i} + t \mathcal{H}_\mathrm{f}$. To do that we will use the Lie-Trotter approximation at the first order.
    
    \begin{proposition}[First-order Lie-Trotter approximation]\label{prop:trotter}
    	Let $A$ and $B$ be two (possibly non-commuting) operators. Then, for a small time step $\Delta t$, the exponential of their sum can be approximated as
    	\begin{equation}
    		\exp\big\{\mathrm{i} A \Delta t + \mathrm{i} B \Delta t\big\} 
    		= \exp\big\{\mathrm{i} A \Delta t\big\} \, \exp\big\{\mathrm{i} B \Delta t\big\} + \mathcal{O}(\Delta t^2).
    	\end{equation}
    \end{proposition}
    
    \begin{proof}
    	A direct expansion shows that:
    	\begin{align}
    		\exp\big\{\mathrm{i} A \Delta t + \mathrm{i} B \Delta t\big\} 
    		&= \mathbb{I} + \mathrm{i} A \Delta t + \mathrm{i} B \Delta t + \mathcal{O}(\Delta t^2), \\
    		\exp\big\{\mathrm{i} A \Delta t\big\} \exp\big\{\mathrm{i} B \Delta t\big\} 
    		&= \left( \mathbb{I} + \mathrm{i} A \Delta t + \mathcal{O}(\Delta t^2) \right)
    		\left( \mathbb{I} + \mathrm{i} B \Delta t + \mathcal{O}(\Delta t^2) \right) \\
    		&= \mathbb{I} + \mathrm{i} A \Delta t + \mathrm{i} B \Delta t + \mathcal{O}(\Delta t^2).
    	\end{align}
    \end{proof}
    
    By choosing $A = - (1-t_k)\mathcal{H}_\mathrm{i}$ and $B = - t_k \mathcal{H}_\mathrm{f}$ and plugging them in Equation (\ref{eq:single-term-trotter}) we obtain 
    \begin{equation}\label{eq:distribute-the-Delta}
    	\lim\limits_{p \to +\infty} 
    	\mathcal{T}
    	\prod\limits
    	_{k=1}^{p} \exp\left\{- \mathrm{i} \Delta t \mathcal{H}(t_k)\right\}
    	=
    	\lim\limits_{p \to +\infty}
    	\mathcal{T}
    	\prod\limits
    	_{k=1}^{p}
    	\left[
    	\exp\left\{-\mathrm{i} (1-t_k)\mathcal{H}_\mathrm{i} \Delta t \right\}
    	\exp\left\{- \mathrm{i} t_k \mathcal{H}_\mathrm{f} \Delta t\right\}
    	+ O(\Delta t^2)
    	\right]
    	= \dots
    \end{equation}
    For notational simplicity, we now define 
    \begin{equation}
    	X_k := \exp\left\{-\mathrm{i} (1-t_k)\mathcal{H}_\mathrm{i} \Delta t \right\}
    	\exp\left\{- \mathrm{i} t_k \mathcal{H}_\mathrm{f} \Delta t\right\}
    \end{equation}
    and we continue from Equation (\ref{eq:distribute-the-Delta}) writing
    \begin{equation}
    	\begin{aligned}
    		\dots = \lim_{p \to +\infty} \mathcal{T} \prod_{k=1}^{p} \left( X_k + O(\Delta t^2) \right)
    		&= \lim_{p \to +\infty} 
    		\mathcal{T}
    		\Bigg[
    		\prod_{k=1}^{p} X_k
    		+ \sum_{j=1}^{p} \left( \prod_{k=j+1}^{p} X_k \right) O(\Delta t^2) \left( \prod_{k=1}^{j-1} X_k \right)
    		\Bigg] \\
    		&= \lim_{p \to +\infty} 
    		\mathcal{T}
    		\prod_{k=1}^{p} X_k + O(p \Delta t^2) \\
    		&\approx 
    		\mathcal{T}
    		\prod_{k=1}^{p} X_k \\
    		&=
    		X_p \dots X_1.
    	\end{aligned}
    \end{equation}
    By plugging in the definition of $X_k$ we obtain
    \begin{equation}\label{eq:final-drop-of-limit}
    	\begin{aligned}
    		\lim\limits_{p \to +\infty} 
    		\mathcal{T}
    		\prod\limits
    		_{k=1}^{p} \exp\left\{- \mathrm{i} \Delta t \mathcal{H}(t_k)\right\}
    		&\approx
    		\mathcal{T}
    		\prod_{k=1}^{p} 
    		\exp\left\{-\mathrm{i} (1-t_k)\mathcal{H}_\mathrm{i} \Delta t \right\}
    		\exp\left\{- \mathrm{i} t_k \mathcal{H}_\mathrm{f} \Delta t\right\}
    		.
    	\end{aligned}
    \end{equation}
    In conclusion, by putting together Equations (\ref{eq:dyson-approx}) and (\ref{eq:final-drop-of-limit}) we obtain the result of Equation (\ref{eq:approximation-of-U}), namely
    \begin{equation}\label{eq:final-approx}
    	U(t) = 
    	\mathcal{T} \exp\!\Bigg(-\mathrm{i} \int_{t_0}^t H(s)\, ds \Bigg) 
    	\approx 
    	\mathcal{T}
    	\prod_{k=1}^{p} 
    	e^{-\mathrm{i} (1-t_k)\mathcal{H}_\mathrm{i} \Delta t}
    	e^{- \mathrm{i} t_k \mathcal{H}_\mathrm{f} \Delta t}
    \end{equation}
    \subsection{PUBO Hamiltonian Coefficients}\label{appendix:calculations}

    We show the steps that lead from the Hamiltonian of Equation (\ref{eq:pubo-with-s+1}), where the change of variable has been performed, to Equation (\ref{eq:ham-with-f}), where the spin terms of the same degree have been collected.

    We start by rewriting the product $(s_{i_1} + 1) \dots (s_{i_k} + 1)$ as 
    \begin{equation}\label{eq:s+1-terms}
        (s_{i_1} + 1) \dots (s_{i_k} + 1) 
        =
        \sum\limits_{h=1}^k     
        \sum\limits_{\substack{S \subseteq I_k \\ |S|=h}}
        \prod\limits_{l \in S} 
        s_l
        + \bcancel{1}
        ,
    \end{equation}
    where $I_k = \{i_1, \dots, i_k\}$ is the set of indices and the 1 is removed since it is a constant. The expression in Equation (\ref{eq:s+1-terms}) yields all the terms with $h=1, \dots, k$ product of spins, each one repeated with different indices taken as a subset of $I_h$. We give a concrete example for $k=3$.
    \begin{equation}
        (s_{i_1} + 1)(s_{i_2} + 1)(s_{i_3} + 1) 
        = 
        \underbrace{s_{i_1} + s_{i_2} + s_{i_3}}_\text{
            $\sum\limits_{\substack{S \subseteq \{i_1,i_2,i_3\} \\ |S| = 1}} \prod\limits_{l \in S} s_l$
        }
        + 
        \underbrace{s_{i_1} s_{i_2} + s_{i_2} s_{i_3} + s_{i_1} s_{i_3}}_\text{
            $\sum\limits_{\substack{S \subseteq \{i_1,i_2,i_3\} \\ |S| = 2}} \prod\limits_{l \in S} s_l$
        } 
        + 
        \underbrace{s_{i_1} s_{i_2} s_{i_3}}_\text{
            $\sum\limits_{\substack{S \subseteq \{i_1,i_2,i_3\} \\ |S| = 3}} \prod\limits_{l \in S} s_l$
        } 
        +
        1
        .
    \end{equation}
    For notational purposes we introduce $q'_k := 2^{-k} q_k$ and $\sigma_{k,h} := \binom{k}{h}$. We plug Equation (\ref{eq:s+1-terms}) into Equation (\ref{eq:pubo-with-s+1}) and expand the sum:
    \begin{equation}\label{eq:comp1}
    \begin{aligned}
        \mathcal{H}(s)
        &=
        \sum\limits_{k=1}^{d} 
        \sum\limits_{i_1 \ldots i_k=1}^n 
        q'_{k,i_1 \ldots i_k}
        \sum\limits_{h=1}^k     
        \sum\limits_{\substack{S \subseteq I_k \\ |S|=h}}
        \prod\limits_{l \in S} 
        s_l
        \\
        &\stackrel{(a)}{=}
        \sum\limits_{k=1}^{d} 
        \sum\limits_{i_1 \ldots i_k=1}^n
        q'_{k,i_1 \ldots i_k}
        \left(
            \sum\limits_{\substack{S \subseteq I_k \\ |S|=1}}
            \prod\limits_{l \in S} 
            s_l
            +
            \dots
            +
            \sum\limits_{\substack{S \subseteq I_k \\ |S|=k}}
            \prod\limits_{l \in S} 
            s_l
        \right)
        \\
        &\stackrel{(b)}{=}
        \sum\limits_{i_1=1}^n   
        \;\; q'_{1,i_1} \;\;\;\;\;\;
        \left(
            \sum\limits_{\substack{S \subseteq I_1 \\ |S|=1}}
            \prod\limits_{l \in S} 
            s_l
        \right)
        +
        \\
        & \;\;\;
        \sum\limits_{i_1 i_2=1}^n   
        \; q'_{2,i_1 i_2} \hspace{9 pt}
        \left(
            \sum\limits_{\substack{S \subseteq I_2 \\ |S|=1}}
            \prod\limits_{l \in S} 
            s_l
            +
            \sum\limits_{\substack{S \subseteq I_2 \\ |S|=2}}
            \prod\limits_{l \in S} 
            s_l
        \right)
        +
        \\
        & \;\;\;\;\;\;\;\; \vdots
        \\
        & \;\;
        \sum\limits_{i_1 \dots i_d=1}^n    
        q'_{d,i_1 \dots i_d}
        \left(
            \sum\limits_{\substack{S \subseteq I_d \\ |S|=1}}
            \prod\limits_{l \in S} 
            s_l
            +
            \sum\limits_{\substack{S \subseteq I_d \\ |S|=2}}
            \prod\limits_{l \in S} 
            +
            \dots
            +
            \sum\limits_{\substack{S \subseteq I_d \\ |S|=d}}
            \prod\limits_{l \in S} 
            s_l
        \right)
        = \ldots
    \end{aligned}
    \end{equation}
    where in (a) we expanded the sum $\sum\limits_{h=1}^k$ and in (b) the sum $\sum\limits_{k=1}^d$. We now consider the fact that, due to the symmetry under permutation of $q_k$, we can redefine the indices of the spin products $\prod\limits_{l \in S} s_l$ so that only ordered sequences appear. As an example on a single term, redefining the indices $i_1 \leftrightarrow i_2$ and using the fact that $q'_{2,i_1, i_2} = q'_{2,i_2, i_1}$ we obtain
    \begin{equation}
        \sum\limits_{i_1 i_2=1}^n q'_{2,i_1 i_2} s_{i_2} 
        =
        \sum\limits_{i_2 i_1=1}^n q'_{2,i_2 i_1} s_{i_1} 
        =
        \sum\limits_{i_1 i_2=1}^n q'_{2,i_1 i_2} s_{i_1} 
        .
    \end{equation} 
    Because of this we can remove the sum $\sum_{\substack{S \subseteq I_k \\ |S|=h}}$ provided that we count all its terms, which are the subsets of $I_k$ of size $h$, namely the combinations $\sigma_{k,h}$. If we substitute the sum with the number of combinations and write only the ordered sequence of spins we can continue from Equation (\ref{eq:comp1}) as
    \begin{equation}
    \begin{aligned}
        \\
        \dots
        &=
        \sum\limits_{i_1=1}^n  
        \;\; q'_{1,i_1} \;\;\;\;\;\;
        \left(
            \sigma_{1,1}
            s_{i_1}
        \right)
        +
        \\
        & \;\;\;
        \sum\limits_{i_1 i_2=1}^n
        \; q'_{2,i_1 i_2} \hspace{9 pt}
        \left(
            \sigma_{2,1}
            s_{i_1}
            +
            \sigma_{2,2}
            s_{i_1}s_{i_2}
        \right)
        +
        \\
        & \;\;\;\;\;\;\;\; \vdots
        \\
        & \;\;
        \sum\limits_{i_1 \dots i_d=1}^n
        q'_{d,i_1 \dots i_d}
        \left(
            \sigma_{d,1}
            s_{i_1}
            +
            \sigma_{d,2}
            s_{i_1} s_{i_2}
            +
            \dots
            +
            \sigma_{d,d}
            s_{i_1} s_{i_2} \ldots s_{i_d}
        \right)
        = 
        \dots
    \end{aligned}
    \end{equation}
    We now distribute the sums and collect the terms column-wise, summing all the coefficients that multiply the same spins. In addition, we collect the sum on the index $i_1$ in the first column, the sum on $i_1, i_2$ in the second, and so on until $i_1, \dots, i_d$. We obtain
    \begin{equation}
    \begin{aligned}
        \dots
        =
        &\sum\limits_{i_1=1}^n
        \hspace{4.5pt}
        \left(
            \sigma_{1,1} q'_{1, i_1} + 
            \sigma_{2,1} \sum\limits_{i_2=1}^n q'_{2,i_1 i_2} +
            \dots +
            \sigma_{d,1} \sum\limits_{i_2 \ldots i_d =1}^n q'_{d,i_1 \dots i_d}
        \right)
        s_{i_1}
        +
        \\
        & \hspace{-3.5pt} \sum\limits_{i_1 i_2=1}^n
        \hspace{1.5pt}
        \left(
            \hspace{67pt}
            \sigma_{2,2} q'_{2,i_1 i_2} +
            \dots +
            \sigma_{d,2} \sum\limits_{i_3 \ldots i_d =1}^n q'_{d,i_1 \dots i_d}
        \right)
        s_{i_1}s_{i_2}
        +
        \\
        & \hspace{9pt} \vdots
        \\
        & \hspace{-8pt} \sum\limits_{i_1 \dots i_d=1}^n 
        \left(
            \hspace{181pt}
            \sigma_{d,d} q'_{k,i_1 \dots i_d}
        \right)
        s_{i_1} \dots s_{i_d}
        \\
        =
        &\sum\limits_{k=1}^d
        \sum\limits_{i_1 \dots i_k = 1}^n
        \left[
            q'_{k, i_1 \dots i_k}
            +
            \sum\limits_{h = 1}^{d-k} 
            \sum\limits_{i_{k+1} \dots i_{k+h} = 1}^n
            \binom{k+h}{k}
            q'_{k, i_1 \dots i_{k+h}}
        \right]
        s_{i_1} \dots s_{i_k}
        ,
    \end{aligned}    
    \end{equation}
    where in the last step we used the fact that $\sigma_{k,k} = 1$. This result corresponds to the one of Equation (\ref{eq:ham-with-f}), provided the definition of $f_{k,i_1,\dots, i_k}$ as the term in square brackets.
    
    \subsection{$R_{Z^k}$ Gate Decomposition}\label{appendix:pubo-qaoa}
    \subsubsection*{$R_{Z^2}$ decomposition} 
        We begin by revisiting the decomposition of \(R_{Z^2}\) from Equation~\eqref{eq:exp-term}, using a different notation that will facilitate the generalization of this calculation to \(R_{Z^k}\). We write
        \begin{equation}\label{eq:Rzz-revisited}
            \begin{aligned}
                \exp\left\{-\mathrm{i} \frac{\theta}{2} \sigma^{i_1}_z \sigma^{i_2}_z \right\} 
                &=
                \exp\left\{
                    -\mathrm{i} \frac{\theta}{2}
                    \text{diag}\left\{1, -1, -1, 1\right\}
                \right\} 
                \\
                &=
                \text{diag}\left\{e^{-\mathrm{i} \theta/2}, e^{\mathrm{i} \theta/2}, e^{\mathrm{i} \theta/2}, e^{-\mathrm{i} \theta/2}\right\}
                \\
                &\stackrel{(a)}{=}
                \text{diag}\left\{R_Z, R_Z^{-1}\right\}
                \\
                &\stackrel{(b)}{=}
                \text{diag}\left\{\mathbb{I}, X\right\}\text{diag}\left\{R_Z, R_Z\right\} \text{diag}\left\{\mathbb{I}, X\right\},
             \end{aligned}
        \end{equation}
    where in (a) we used the fact that $R_Z := R_Z(\theta) = \text{diag}\left\{ e^{-\mathrm{i} \theta/2}, e^{\mathrm{i} \theta/2} \right\}$, while in (b) that $R_Z^{-1} := R_Z(-\theta) = X R_Z X $, as can be explicitly computed. We conclude by noticing that 
    \begin{equation}\label{eq:finding-the-cnot}
    \begin{aligned}
        &
        \text{diag}\left\{R_Z, R_Z\right\}
        =
        \mathbb{I} \otimes R_Z
        \\
        &
        \text{diag}\left\{\mathbb{I}, X\right\}
        =
        \ket{0}\bra{0}_{i_1} \otimes \mathbb{I}_{i_2} + \ket{1}\bra{1}_{i_1} \otimes X_{i_2} = \text{CNOT}(i_1, i_2)
        ,
    \end{aligned}
    \end{equation}
    which can be computed directly.
    Plugging Equation (\ref{eq:finding-the-cnot}) into Equation (\ref{eq:Rzz-revisited}) we obtain the same results of Equation (\ref{eq:exp-term}).
    \subsubsection*{$R_{Z^k}$ decomposition} 
        We now generalize this calculation to $R_{Z^k}$, which is given by
        \begin{equation}\label{eq:Rzk-1}
            \begin{aligned}
                R_{Z^k}
                =
                \exp\left\{-\mathrm{i} \frac{\theta}{2} \sigma^{i_1}_z \otimes \dots \otimes \sigma^{i_k}_z \right\} 
             \end{aligned}
        \end{equation}
        We first explicitly compute the term $\sigma^{i_1}_z \otimes \dots \otimes \sigma^{i_k}_z$. This will give a diagonal matrix with values $1$ and $-1$. Its explicit form can be found by writing each $\sigma_z$ as 
        \begin{equation}
            \sigma_z = \sum\limits_{i=0}^1 (-1)^i \ket{i} \bra{i}
            .
        \end{equation}     
        We can then expand $\sigma_z^{\otimes k}$
        \begin{equation}\label{eq:sigma1-tensor-sigmak-pt1}
            \begin{aligned}
                \sigma^{i_1}_z \otimes \dots \otimes \sigma^{i_k}_z 
                &=
                \sum\limits_{i_1=0}^1 (-1)^{i_1} \ket{i_1}\bra{i_1} \otimes \dots \otimes \sum\limits_{i_k=0}^1 (-1)^{i_k} \ket{i_k}\bra{i_k}
                \\
                &=
                \sum\limits_{i_1 \dots i_k = 0}^1 (-1)^{i_1 + \dots +i_k} \ket{i_1}\bra{i_1} \otimes \dots \otimes \ket{i_k}\bra{i_k} 
                \\
                &=
                \text{diag}
                \{
                    (-1)^{0 + \dots + 0}, 
                    (-1)^{0 + \dots + 1},
                    \dots,
                    (-1)^{1 + \dots + 1}
                \}
                .
             \end{aligned}
        \end{equation}
        The $i$-th entry of the above diagonal matrix, which is the $i$-th element of the sum, is given by
        \begin{equation}
            (-1)^{i_1 + \dots + i_k} \ket{i_1}\bra{i_1}\otimes \dots \otimes \ket{i_k}\bra{i_k}
            ,
        \end{equation}
        where $i_1, \dots, i_k$ are the digits of the binary representation of $i$. The exponent of the coefficient, namely $i_1 + \dots + i_k$, counts the number of ones in the binary representation of $i$. This is known as the Hamming weight of $i$, which we denote by $h(i) := \sum\limits_{l = 1}^k i_l$.

        In conclusion we can write 
        \begin{equation}\label{eq:sigma1-tensor-sigmak-pt3}
            \begin{aligned}
                \sigma^{i_1}_z \otimes \dots \otimes \sigma^{i_k}_z 
                =
                \text{diag}
                \{
                    (-1)^{h(0)}, 
                    (-1)^{h(1)},
                    \dots,
                    (-1)^{h(2^{k}-1)}
                \}
                ,
             \end{aligned}
        \end{equation}
        where $(-1)^{h(i)} = 1$ if $h(i)$ is even and $(-1)^{h(i)} = -1$ if $h(i)$ is odd. We can also rewrite $\sigma^{\otimes k}_z$ using Equation (\ref{eq:sigma1-tensor-sigmak-pt3}) only on the first $k-1$ sigmas. Namely
        \begin{equation}\label{eq:sigma1-tensor-sigmak-pt4}
            \begin{aligned}
                \sigma^{\otimes k}_z
                &=
                \sigma^{\otimes k-1}_z \otimes \sigma_z 
                \\
                &=
                \text{diag}
                \{
                    (-1)^{h(0)}, 
                    (-1)^{h(1)},
                    \dots,
                    (-1)^{h(2^{k-1}-1)}
                \}
                \otimes
                \text{diag}\{1, -1\}
                \\
                &=
                \text{diag}
                \{
                    (-1)^{h(0)}, (-1)^{h(0)+1}, 
                    (-1)^{h(1)}, (-1)^{h(1)+1},
                    \dots, 
                    (-1)^{h(2^{k-1}-1)}, (-1)^{h(2^{k-1}-1)+1}
                \}
                ,
             \end{aligned}
        \end{equation}
        which will be a more useful form for $\sigma^{\otimes k}_z$ in the following. To compute $R_{Z^k}$ we start from Equation (\ref{eq:Rzk-1}) and write 
        \begin{equation}\label{eq:Rzk-3}
            \begin{aligned}
                R_{Z^k}
                &=
                \exp\left\{-\mathrm{i} \frac{\theta}{2} \sigma^{i_1}_z \otimes \dots \otimes \sigma^{i_k}_z \right\} 
                \\
                &=
                \exp\left\{
                   -\mathrm{i} \frac{\theta}{2}
                   \text{diag}
                   \{
                    (-1)^{h(0)}, (-1)^{h(0)+1}, 
                    \dots, 
                    (-1)^{h(2^{k-1}-1)}, (-1)^{h(2^{k-1}-1)+1}
                   \}                
                \right\}
                \\
                &=
                \text{diag}
                \left\{
                    e^{-\mathrm{i} \frac{\theta}{2} (-1)^{h(0)}},
                    e^{\mathrm{i} \frac{\theta}{2} (-1)^{h(0)}},
                    \dots, 
                    e^{-\mathrm{i} \frac{\theta}{2} (-1)^{h(2^{k-1}-1)}},
                    e^{\mathrm{i} \frac{\theta}{2} (-1)^{h(2^{k-1}-1)}}
               \right\} 
                \\
                &=
                \text{diag}
                \left\{
                    R_Z^{(-1)^{h(0)}}, 
                    \dots,
                    R_Z^{(-1)^{h(2^{k-1}-1)}}
               \right\} 
               = 
               \dots
        \end{aligned}
        \end{equation}
    where if $h(i)$ is even then $R_Z^{(-1)^{h(i)}} = R_Z$; if $h(i)$ is odd then $R_Z^{(-1)^{h(i)}} = R_Z^{-1}$. We now remember that $R_Z^{-1} = X R_Z X$, which we can use as we did in equality (b) of Equation (\ref{eq:Rzz-revisited}). If, then, $h(i)$ is odd we write $R_Z^{(-1)^{h(i)}} =  X R_Z X$, if it is even we write $R_Z^{(-1)^{h(i)}} =  \mathbb{I} R_Z \mathbb{I}$. In general this can be written as 
    \begin{equation}
        R_Z^{(-1)^{h(i)}} =  X^{\sigma(i)} R_Z X^{\sigma(i)}, 
        \text{    with }
        \sigma(i) = 
        \begin{cases}
            1 \text{ if } h(i) \text{ is odd}\\ 
            0 \text{ if } h(i) \text{ is even}\\ 
        \end{cases}
        .
    \end{equation}
    With this notation, we can continue from Equation (\ref{eq:Rzk-3}) as
    \begin{equation}\label{eq:Rzk-4}
            \begin{aligned}
                \dots
                &=
                \text{diag}
                \left\{
                    X^{\sigma(0)}, 
                    \dots,
                    X^{\sigma(2^{k-1}-1)}
               \right\} 
                \text{diag}
                \left\{
                    R_Z, 
                    \dots,
                    R_Z
               \right\} 
               \text{diag}
                \left\{
                    X^{\sigma(0)}, 
                    \dots,
                    X^{\sigma(2^{k-1}-1)}
               \right\} 
               .
        \end{aligned}
        \end{equation}
    Analogously to Equation (\ref{eq:finding-the-cnot}) we now proceed to further decompose these block-diagonal matrices, namely
    \begin{equation}\label{eq:almost-decomposed-rzk}
        \text{diag}\left\{R_Z, \dots, R_Z\right\}
        =
        \mathbb{I} \otimes \dots \otimes \mathbb{I} \otimes R_Z
        ,
    \end{equation}
    where the identities act on qubits $1, \dots, k-1$ while $R_Z$ acts on qubit $k$. 
    The remaining terms on the left- and right-hand sides of Equation (\ref{eq:almost-decomposed-rzk}) can instead be decomposed as
    \begin{equation}\label{eq:finding-the-cnot-pubo2}
    \begin{aligned}
        \text{diag}
            \left\{
                X^{\sigma(0)}, 
                \dots,
                X^{\sigma(2^{k-1}-1)}
            \right\} 
        &=
        \sum\limits_{i=0}^{2^{k-1}-1}
        \ket{i}\bra{i} \otimes X^{\sigma(i)}
        =
        \dots
    \end{aligned}
    \end{equation}
    where we stress the fact that in Equation (\ref{eq:finding-the-cnot-pubo2}) the projector $\ket{i}\bra{i}$ is on the first $k-1$ qubits, while the operator $X^{\sigma(i)}$ acts on the last one, namely the qubit $k$. We continue from Equation (\ref{eq:finding-the-cnot-pubo2}) as
    \begin{equation}
    \begin{aligned}
        \dots
        &=
        \sum\limits_{i_1, \dots, i_{k-1} = 0 }^1
        \ket{i_1, \dots, i_{k-1}}\bra{i_1, \dots, i_{k-1}} \otimes X^{\sigma(i)}_{i_k}
        \\
        &=
        \sum\limits_{i_1, \dots, i_{k-1} = 0 }^1
        \ket{i_1}\bra{i_1} \otimes \dots \otimes \ket{i_{k-1}}\bra{i_{k-1}} \otimes X^{\sigma(i)}_{i_{k}}
        \\
        &\stackrel{(a)}{=}
        \left(
            \ket{0}\bra{0}_{i_1} \otimes \mathbb{I}_{i_{k}} + \ket{1}\bra{1}_{i_1} \otimes X_{i_{k}}
        \right)
        \dots
        \left(
            \ket{0}\bra{0}_{i_{k-1}} \otimes \mathbb{I}_{i_{k}} + \ket{1}\bra{1}_{i_{k-1}} \otimes X_{i_{k}}
        \right)
        \\
        &=
        \text{CNOT}(i_1, i_{k}) \dots \text{CNOT}(i_{k-1}, i_{k})
        .
    \end{aligned}
    \end{equation}
    Indeed, we can see on the right hand side of equality (a) that, once we computed all the products, if we select a particular state in the sum, for every $\ket{1}\bra{1}$ there is a $X_{i_k}$ operator. If there is an even number of $\ket{1}\bra{1}$, then we get $X_{i_k}^{2n} = \mathbb{I}$, else $X_{i_k}^{2n+1} = X_{i_k}$, which is the result expressed as $X^{\sigma(i)}_{i_k}$.

    \bibliography{refs}

\end{document}